\documentclass[sigconf,nonacm]{acmart}

\usepackage{amsmath}
\usepackage{algorithmic}
\usepackage{graphicx}
\usepackage{textcomp}
\usepackage{xcolor}

\usepackage{pgfplots}
\usepackage{pgfplotstable}
\pgfplotsset{compat=newest}
\usepackage{algorithm}
\usepackage{makecell}
\usepackage{enumitem}
\usepackage{multirow}
\usepackage{hhline}

\usepackage{tikz}
\usetikzlibrary{arrows,shapes,trees}
\usetikzlibrary{arrows.meta}

\title{Replication in Graph Partitioning and Scheduling Problems}

\author{P\'al Andr\'as Papp}
\orcid{0009-0005-6667-802X}
\email{pal.andras.papp@huawei.com}
\affiliation{
  \department{Computing Systems Lab}
  \institution{Huawei Zurich Research Center}
  \city{Zurich}
  \country{Switzerland}
}

\author{Toni B\"ohnlein}
\email{toni.boehnlein@huawei.com}
\orcid{0009-0001-2152-022X}
\affiliation{
  \department{Computing Systems Lab}
  \institution{Huawei Zurich Research Center}
  \city{Zurich}
  \country{Switzerland}
}

\author{Albert-Jan N. Yzelman}
\email{albertjan.yzelman@huawei.com}
\orcid{0000-0001-8842-3689}
\affiliation{
 \department{Computing Systems Lab}
  \institution{Huawei Zurich Research Center}
  \city{Zurich}
  \country{Switzerland}
}

\begin{document}

\begin{abstract}
The efficient parallel execution of complex computations requires balancing the workload across processors while minimizing the communication between them. This inherent trade-off is often captured by graph partitioning or DAG scheduling problems. For the sake of model simplicity, most works on these problems assume that nodes can be assigned to only a single processor. However, in reality, \emph{replicating} an operation on several processors can easily be beneficial: it may increase the computational costs only by a small amount, while significantly reducing the communication requirements.

Our goal is to provide a comprehensive analysis of the impact of replication on partitioning and scheduling problems. On the theoretical side, we show that for graph partitioning, replication makes the problem significantly harder in terms of approximation complexity, whereas for scheduling, its impact on complexity seems less prominent. On the experimental side, we conduct a thorough analysis of the cost reduction obtainable by replication, on a wide range of graphs from real-world applications. For hypergraph partitioning, we use Integer Linear Programming (ILP) formulations to compare the optimal costs; our experiments show that replication can reduce the cost by $17\%$--$65\%$ on average, or even entirely remove the need for communication in some cases. For DAG scheduling, we similarly use ILPs on smaller graphs, and develop a sophisticated heuristic that is also applicable to much larger workloads. Our experiments here demonstrate a mean cost reduction of $11.61\%$--$23.13\%$ with replication, or even up to $58.17\%$ in some cases.
\end{abstract}

\ccsdesc[300]{Theory of computation~Parallel computing models}

\keywords{Partitioning, Scheduling, Hypergraph, DAG, Replication, Duplication}

\maketitle

\section{Introduction}

The efficient parallel execution of complex computations on multiple devices is a fundamental problem in computer science, and a crucial task in various domains like scientific computing, machine learning or big data. Finding the best parallel execution strategy is often a challenging problem with contradicting objectives, e.g.\ distributing the computational workload in a balanced way but also minimizing the costly communication steps between the devices.

In mathematical terms, these tasks are often best captured as a (hyper-)graph partitioning or scheduling problem. The nodes of the graph represent the operations or subtasks in the computation that need to be executed, and the (hyper-)edges represent data dependencies among specific operations. The goal is to find an assignment of the nodes to the available devices which has an acceptable load balance (in terms of e.g.\ computational workload, or memory usage), and also minimizes the communication between the devices.

Beyond this general idea, the concrete details of the model can vary significantly depending on the application. One prominent setting is graph or hypergraph partitioning, where the goal is simply to find an assignment of the nodes to devices. Here, load balancing is typically handled via an explicit balance constraint, i.e.\ an upper bound on the number of nodes in each partition. The objective is to minimize the number of cut (hyper-)edges, i.e.\ those that span multiple partitions. This offers a simple and convenient formulation in many settings where the time dimension of the execution can be ignored due to some modeling choice or external factor.

Another related problem is scheduling, where the computation is often captured as a directed acyclic graph (DAG), with the directed edges representing data dependencies between operations. Scheduling requires us to not only assign each node to a device, but also to a concrete time step. This is indeed crucial in order to capture whether two operations can be executed in parallel when we have precedence constraints. Given an assignment to devices and time steps, a scheduling model defines the time when each operation is finished, based on its predecessors and the potential communication steps needed. This way, a single objective function allows to implicitly combine the goals of load balancing and minimizing communication. DAG scheduling is a general problem model that can capture any static computational task.

The main focus of our paper is the concept of \emph{replication} in these problems, also known as duplication or re-computation in some cases. For simplicity, many works on these problems assume that every node is assigned to a single device only. In contrast, replication means that we are also allowed to assign a node to multiple devices, i.e.\ redundantly execute the same operation on more than one device. This increases the use of compute resources, but on the other hand, it can also significantly decrease the amount of communication required, and hence result in a lower cost altogether. As such, while forbidding replication leads to more convenient mathematical problems, replication is in fact an essential ingredient to realistically capture the optimal costs in these problems.

Replication has been considered many times both in theoretical works and in applications, but we are not aware of any comprehensive study to analyze its impact on partitioning and scheduling problems. Apart from some worst-case examples and lower bounds, many of the fundamental questions remained open. How does replication affect the computational complexity of partitioning or scheduling? How much can it decrease the cost of the optimal solution in more practical cases, i.e.\ benchmarks that describe real-world computations? Are there specific kinds of graphs or applications where replication is particularly useful for reducing costs?

Our main goal in this paper is to understand how replication impacts partitioning and scheduling problems in general. We first define the extension of the hypergraph partitioning and DAG scheduling problems with replication, and show simple examples where replication can decrease the optimal cost by a large factor. We then analyze how replication affects the computational complexity of these problems. In hypergraph partitioning, the problem becomes significantly more challenging: we prove that the optimal solution with replication is NP-hard to approximate to any finite factor. In contrast, in scheduling, our observations suggest that replication has a smaller impact on the problem complexity.

Our main contribution is a comprehensive experimental study of replication in these problems. In hypergraph partitioning, we discuss two ways to formulate partitioning with replication as an Integer Linear Program (ILP). We use a commercial ILP solver to find (close to) optimal solutions to the problem over a diverse hypergraph benchmark that we derive from Mixture of Experts (MoE) model distribution and sparse matrix--vector multiplication (SpMV). We compare the solutions to the optimal non-replicating partitioning on the same problems; our experiments show that replication can reduce the cost by $17\%$--$65\%$ on different datasets, or even entirely remove the need for communication in a few cases.

In DAG scheduling, we focus on the Bulk Synchronous Parallel (BSP) model, a well-established model in parallel computing~\cite{BSPintro}. We also consider an ILP-based approach here, but for scheduling, this is only viable for very small graphs. Hence we also consider a more heuristic method: we develop a complex local-search-like algorithm to introduce replication steps into a schedule whenever this decreases the cost altogether. Our results show that this replicating algorithm achieves a $11.61\%$--$23.13\%$ mean cost reduction compared to an already strong baseline scheduler, or up to a mean of $58.17\%$ for larger communication cost parameters.

\section{Related Work}

Partitioning and scheduling problems have both been extensively studied before. In Section~\ref{sec:model}, we also discuss some concrete use cases and elaborate on the difference between the two settings.

Balanced graph partitioning is a fundamental NP-hard problem that has applications ranging from VLSI design to scientific computing~\cite{HGapplic1}. There is specific line of work on approximation algorithms for the perfectly balanced case (bisection problem), resulting in a $O(\log n)$-factor approximation~\cite{Feige00, Racke08}. Others considered bi-criteria approximations of the same problem where the balance constraint can also be violated to some factor~\cite{Leighton99, Krauthgamer09}. There are similar works for hypergraph partitioning, including both approximation algorithms and hardness results~\cite{BisectionApprox, hypergraphs_opdas}. There are also several works on hierarchical variants of the problems~\cite{Haji14, Racke16, hypergraphs_opdas}.

Due to its practical applications, there are several algorithms and open-source tools to solve the problem in practice, including both exact algorithms~\cite{Pelt15, knigge20} and heuristics~\cite{hMetis, KaHyPar, Parkway, ccatalyurek2023more}.

DAG scheduling has also been studied in various models, ranging from very simple to more realistic. The simplest model assumes free communication between processors~\cite{DAG2proc1, approx1, approx2}; replication offers no advantage in this setting. A somewhat more complex model assumes fixed communication delays between processors~\cite{commDdef1,commDdef2}; here replication can already be beneficial. There are numerous algorithms and hardness results for this model~\cite{commDunit1, commDappR1, recent1}.

There are several more realistic models inspired from practice as well~\cite{multiBSP1, LogP}; however, these are often too complex with many parameters, and hence are almost exclusively studied for concrete computations (i.e.\ specific, very structured DAGs). One prominent example for a middle ground is the Bulk Synchronous Parallel (BSP) model~\cite{BSPintro}, which is accurate enough to be widely used in libraries and applications~\cite{BSPlib, BSPimpl1}, but still simple enough to allow a thorough analysis of its theoretical properties~\cite{BSPbook2, BSP_DAG_opdas} and the optimal schedules for crucial computations~\cite{mccoll24memory}.

Using replication to save communication cost is a very natural idea and has been explored in a range of problem variants and specific applications~\cite{casas2017balanced,hu2019throughput,bozdag2008compaction}. However, we are not aware of any comprehensive study to analyze the impact of replication on general graphs, either with regard to theoretical properties  or experimentally on a range of computational graphs. Specifically, for (hyper-)graph partitioning, to our knowledge, there is no prior work that explicitly considers the extension of the problem with replication, and compares this to the original problem.

For DAG scheduling, there are many works that allow replication, and several of them directly discuss its impact on the problem; however, these exhibit key differences from ours. The paper of Maiti et al.~\cite{recent4} conducts a theoretical analysis of the impact of replication in the communication delay model: they show that the factor of difference between the optimum with and without replication is upper bounded by $O(\log^2{n} \log P)$, and lower bounded by $\Omega(\log n / \log \log n)$, where $n$ and $P$ are the number of nodes and processors, respectively. Ahmad et al.~\cite{ahmad1998exploiting} study heuristics to introduce replications in the communication delay model, assuming $P=\infty$; while they discuss an idea to adapt their method to finite $P$, their experiments do not study the cost reduction for this case. Papp et al.~\cite{multiBSP_opdas} consider a problem variant with I/O costs (red-blue pebbling), and briefly discuss the empirical cost reduction achieved by replication in this model. The work of~\cite{bozdag2008compaction} studies how to remove replications from a schedule in the communication delay model, in order to decrease the number of processors used. There are also several scheduling papers that do allow replication in their model, but do not study its impact explicitly~\cite{commDappR0, commDappR2,munier1997using}.

\section{Background \& problem definitions} \label{sec:model}

\subsection{Preliminaries}

In our problems, a computation is modeled either as a simple graph, a hypergraph, or a directed acyclic graph (DAG) $G=(V,E)$, where $V$ is the set of nodes, and $E$ is the set of (hyper-)edges. In a simple graph / hypergraph / DAG, each $e \! \in \! E$ is an unordered pair / subset / ordered pair of nodes from $V$, respectively. We will use $n=|V|$. In a DAG, directed cycles are also forbidden, and we call the set of nodes with an edge to / from $v \! \in \! V$ the parents / children of $v$. In a hypergraph, a node-hyperedge pair $(v,e)$ with $v \! \in \! e$ is called a pin, and we denote the total number of pins by $\kappa$.

A more realistic setting may also have node weights $\omega \! : \! V \! \rightarrow \! \mathbb{R}^+$ to express the compute cost of each operation. For a subset $V' \subseteq V$, we will use $\omega(V') = \sum_{v \in V'} \omega(v)$. Similarly, we may have weights on the hyperedges $\mu \! : \! E \! \rightarrow \! \mathbb{R}^+$ (in hypergraph partitioning) or on the nodes $\mu \! : \! V \! \rightarrow \! \mathbb{R}^+$ (in DAG scheduling) to express the size of the communicated data. 
The default problem corresponds to all weights being uniformly $1$. In our theoretical work, we consider the unweighted problem in order to derive stronger hardness results; however, we do use weights in the benchmarks for our experiments.

Another crucial parameter is $P$, the number of available devices (e.g.\ processors or cores) to execute the operations. Throughout the paper, we refer to these as processors for simplicity, and we identify them with the indices $\{ 1, ..., P\}$. We also use the shorthand notation $[P] = \{ 1, ..., P\}$. We assume that $P$ is a small constant; indeed, in practice it is typically significantly smaller than $n$.

\subsection{Graph and Hypergraph Partitioning} \label{sec:def_part}

Partitioning problems assign nodes only to processors and not to time steps. Balanced (hyper-)graph partitioning can capture various parallelization tasks without a time dimension, e.g., streaming settings where a continuous flow of data drives computation.

A prime example is the deployment of modern Large Language Models (LLMs), which process a continuous token stream and need to be distributed over multiple devices due to their massive memory footprint. In the recently prominent Mixture of Experts (MoE) LLMs, the optimal distribution of experts to devices can be naturally captured as a partitioning problem based on the co-occurrence pattern of experts~\cite{moetuner,moe_data}. In a hypergraph representation, the nodes represent concrete experts (in one or more layers), and hyperedges are formed from groups of experts that are frequently invoked together. Hypergraph partitioning corresponds to a balanced distribution of experts (respecting memory constraints on each device) that minimizes the inter-device communication while serving requests.

Similarly, the parallelization of sparse matrix--vector multiplication (SpMV) is often modeled as a hypergraph partitioning problem. For instance, in the so-called fine-grained model, each non-zero entry of the matrix becomes a node, and the entries of each column and row form a separate hyperedge. Intuitively, the nodes here represent the multiplication operations with the given non-zero entry, and vertical (respectively, horizontal) hyperedges represent the values of the input (output) vector that need to be used for (aggregated from) the corresponding multiplications~\cite{knigge20, jenneskens22}.

We illustrate the hypergraph partitioning problem in Figure~\ref{fig:hypergraph}.

\subsubsection*{Without replication}

In balanced hypergraph partitioning, the goal is to partition the nodes of a hypergraph into subsets $V_1$, ..., $V_P$ such that $\bigcup_{i \in [P]} V_i = V$. In the base case without replication, we require the partitions to be disjoint, i.e.\ $V_i \cap V_j = \emptyset$ for $i \neq j$.

Given a real parameter $\varepsilon \! \in \! (0, 1)$, we say that a partitioning is balanced if for all $i \! \in \! [P]$, we have $|V_i| \leq \frac{1+\varepsilon}{P} \cdot n$. With node weights, this condition can be replaced by $\omega(V_i) \leq \frac{1+\varepsilon}{P} \cdot \omega(V)$. This balance constraint has a natural interpretation in most applications: it ensures that each processor is assigned approximately the same workload, in order to guarantee efficient parallelization.

Given this balance constraint, the goal of hypergraph partitioning is to ensure that as few edges are cut as possible. In particular, let $\lambda_e$ denote the number of partitions that a hyperedge $e \! \in \! E$ intersects, i.e.\ $\lambda_e = |\{ i \! \in \! [P] \, | \, e \cap V_i \neq \emptyset \}|$. The cost incurred by hyperedge $e$ is then defined simply as $(\lambda_e-1)$, or $\mu(e) \cdot (\lambda_e-1)$ in case of hyperedge weights. The cost of the entire partitioning is obtained by summing up this cost over all hyperedges $e \! \in \! E$.

This $(\lambda_e-1)$-based cost function is a popular choice because it accurately captures the communication costs in many cases. Generally speaking, if the hyperedge represents a piece of data that is initially associated with one of the $\lambda_e$ processors appearing in the hyperedge, then it takes at least $(\lambda_e-1)$ pairwise communications to share this data among all the processors. For instance, in the MoE hypergraphs above, if the KV cache for a request is on one of the $\lambda_e$ devices, then the corresponding token needs to be sent to $(\lambda_e-1)$ other devices in the given layer. In fine-grained SpMV hypergraphs, if an entry of the input vector (respectively, output vector) is originally stored on (aggregated on) one of the $\lambda_e$ processors appearing in the column (row), then this processor needs to send the value to (needs to receive values from) $(\lambda_e-1)$ others during the computation.

Note that in some use cases, a cost of $\lambda_e$ for each hyperedge $e$ might be more realistic. However, this leads to the same optimal solution(s), as it only adds a fixed term of $|E|$ or $\sum_e \mu(e)$ in total.

In the balanced hypergraph partitioning problem (for a fixed $P$ and $\varepsilon$), our input is a hypergraph, and our goal is to find a balanced partitioning of the hypergraph that minimizes the total cost. Regular graph partitioning, a well-studied problem itself, can be obtained as a special case where each hyperedge in the input has size $2$.

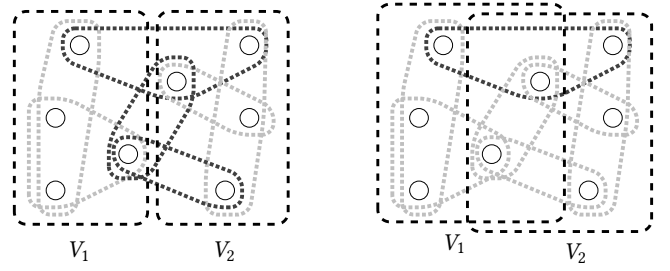
\begin{figure}[t]
    \centering
     \resizebox{0.48\textwidth}{!}{\begin{tikzpicture}

	\begin{scope}[ultra thick, densely dotted, rounded corners = 5pt]
	\draw[lightgray] (-11pt,-7pt) -- (-11pt,38pt) -- (5pt,38pt) -- (37pt,21pt) -- (37pt,9pt) -- (5pt,-7pt) -- cycle;
    \draw[lightgray] (-7pt,-10pt) -- (-7pt,40pt) -- (0pt,70pt) -- (18pt,70pt) -- (18pt,52pt) -- (8pt,-10pt) -- cycle;
    \draw[darkgray] (3pt,67pt) -- (87pt,67pt) -- (87pt,54pt) -- (56pt,38pt) -- (44pt,38pt) -- (3pt,54pt) -- cycle;
    \draw[lightgray] (87pt,70pt) -- (86pt,30pt) -- (80pt,-10pt) -- (62pt,-10pt) -- (62pt,8pt) -- (72pt,70pt) -- cycle;
    \draw[darkgray] (77pt,-7pt) -- (65pt,-7pt) -- (24pt,10pt) -- (24pt,22pt) -- (35pt,22pt) -- (77pt,6pt) -- cycle;
    \draw[lightgray] (90pt,23pt) -- (75pt,23pt) -- (43pt,40pt) -- (43pt,52pt) -- (54pt,52pt) -- (90pt,34pt) -- cycle;
    \draw[darkgray] (22pt,5pt) -- (22pt,21pt) -- (43pt,55pt) -- (57pt,55pt) -- (57pt,37pt) -- (38pt,5pt) -- cycle;

	\end{scope}

    \begin{scope}[very thick, dashed, rounded corners = 5pt]
    \draw (-17pt,-14pt) -- (-17pt,74pt) -- (38pt,74pt) -- (38pt,-14pt) -- cycle;
    \draw (96pt,-14pt) -- (96pt,74pt) -- (42pt,74pt) -- (42pt,-14pt) -- cycle;
    \end{scope}

    \node[anchor=center] at (10pt,-25pt) {\large $V_1$};
    \node[anchor=center] at (70pt,-25pt) {\large $V_2$};

	\draw[black, fill=white] (10pt,60pt) circle (1.0ex);
	\draw[black, fill=white] (0pt,30pt) circle (1.0ex);
	\draw[black, fill=white] (0pt,0pt) circle (1.0ex);
	\draw[black, fill=white] (50pt,45pt) circle (1.0ex);
	\draw[black, fill=white] (30pt,15pt) circle (1.0ex);
	\draw[black, fill=white] (80pt,60pt) circle (1.0ex);
	\draw[black, fill=white] (80pt,30pt) circle (1.0ex);
	\draw[black, fill=white] (70pt,0pt) circle (1.0ex);


	\begin{scope}[ultra thick, densely dotted, rounded corners = 5pt]
	\draw[lightgray] (139pt,-7pt) -- (139pt,38pt) -- (155pt,38pt) -- (187pt,21pt) -- (187pt,9pt) -- (155pt,-7pt) -- cycle;
    \draw[lightgray] (143pt,-10pt) -- (143pt,40pt) -- (150pt,70pt) -- (168pt,70pt) -- (168pt,52pt) -- (158pt,-10pt) -- cycle;
    \draw[darkgray] (153pt,67pt) -- (237pt,67pt) -- (237pt,54pt) -- (206pt,38pt) -- (194pt,38pt) -- (153pt,54pt) -- cycle;
    \draw[lightgray] (237pt,70pt) -- (236pt,30pt) -- (230pt,-10pt) -- (212pt,-10pt) -- (212pt,8pt) -- (222pt,70pt) -- cycle;
    \draw[lightgray] (227pt,-7pt) -- (215pt,-7pt) -- (174pt,10pt) -- (174pt,22pt) -- (185pt,22pt) -- (227pt,6pt) -- cycle;
    \draw[lightgray] (240pt,23pt) -- (225pt,23pt) -- (193pt,40pt) -- (193pt,52pt) -- (204pt,52pt) -- (240pt,34pt) -- cycle;
    \draw[lightgray] (172pt,5pt) -- (172pt,21pt) -- (193pt,55pt) -- (207pt,55pt) -- (207pt,37pt) -- (188pt,5pt) -- cycle;

	\end{scope}

    \begin{scope}[very thick, dashed, rounded corners = 5pt]
    \draw (133pt,-13pt) -- (133pt,77pt) -- (210pt,77pt) -- (210pt,-13pt) -- cycle;
    \draw (246pt,-17pt) -- (246pt,73pt) -- (170pt,73pt) -- (170pt,-17pt) -- cycle;
    \end{scope}

    \node[anchor=center] at (165pt,-22pt) {\large $V_1$};
    \node[anchor=center] at (215pt,-26pt) {\large $V_2$};

	\draw[black, fill=white] (160pt,60pt) circle (1.0ex);
	\draw[black, fill=white] (150pt,30pt) circle (1.0ex);
	\draw[black, fill=white] (150pt,0pt) circle (1.0ex);
	\draw[black, fill=white] (200pt,45pt) circle (1.0ex);
	\draw[black, fill=white] (180pt,15pt) circle (1.0ex);
	\draw[black, fill=white] (230pt,60pt) circle (1.0ex);
	\draw[black, fill=white] (230pt,30pt) circle (1.0ex);
	\draw[black, fill=white] (220pt,0pt) circle (1.0ex);

\end{tikzpicture}}
     \vspace{-16pt}
    \caption{Illustration of hypergraph bipartitioning into two sets $V_1$ and $V_2$, without replication (left) and with replication (right). Cut hyperedges are shown with a darker grey tone.}
    \label{fig:hypergraph}
\end{figure}

\subsubsection*{With replication}

While we are not aware if the extension of graph partitioning with replication has been explicitly studied before, it yields a natural mathematical problem. Assume that $V_1$, ..., $V_P$ are now not required to be disjoint, but must still satisfy the balance constraint; this allows some nodes to be assigned to multiple partitions. For each (hyper-)edge, we want to find the minimal number of partitions that cover all its nodes, in order to again obtain an indicator of the communication costs.

Specifically, let us now redefine $\lambda_e$ as the minimal integer such that there is a subset of processors $\Pi \! \subseteq \! [P]$ such that $|\Pi| = \lambda_e$, and for all $v \! \in \! e$ there is a $p \! \in \! \Pi$ with $v \! \in \! V_p$. For example, if we have a hyperedge $e=(u, v, w)$ such that $u$ is in both $V_1$ and $V_2$, $v$ is in both $V_2$ and $V_3$, and $w$ is in both $V_3$ and $V_4$, then $\lambda_e=2$, because e.g. $V_1$ and $V_3$ together already cover the hyperedge. For the special case of simple graphs and edges $(u,v)$, we simply have $\lambda_e=1$ if $u$ and $v$ share a processor, and $\lambda_e=2$ otherwise. We again sum up $(\lambda_e-1)$ for all $e \! \in \! E$ to obtain the total cost of a partitioning.

This new definition of $\lambda_e$ generalizes the previous cost function to partitioning with replication, again capturing the required number of pairwise communications in many of the use cases. In our MoE hypergraphs, each expert now may be available on multiple devices; the new $(\lambda_e-1)$ describes the minimal number of devices where we need to send the token so that it reaches all experts in hyperedge $e$. In a horizontal hyperedge of the SpMV example (or any other setting where hyperedges describe an aggregation of outputs), each operation may be executed on multiple processors, so the number of communications required to collect the outputs to a single processor is again the new $(\lambda_e-1)$.

We note that a significant drawback of this cost function is that $\lambda_e$ is rather challenging to compute, as it actually forms an instance of the set cover problem. However, when $P$ is a small constant, this is not a problem in terms of complexity or practical viability.

\subsection{DAG scheduling}

Consider a general static computation with precedence constraints between the operations, which can be naturally captured as a DAG. It is easy to see that a single balance constraint per processor does not ensure parallelization here: we might be forced to execute many or even all the operations of one processor before another processor can start working. Some ideas have been explored to extend partitioning to this setting, like layer-wise partitioning the DAG, but these have strong limitations~\cite{hypergraphs_opdas}.

The only model to capture parallelization accurately in this case is scheduling, where we assign the nodes not only to processors, but also to time steps. In contrast to partitioning, the DAG scheduling problem offers a general model that can be applied to any static computation; however, it is also often more challenging to solve or analyze due to the extra time dimension.

DAG scheduling has been exhaustively studied in many models. Replication only becomes relevant in more realistic models that account for communication costs between processors, such as the communication delay model~\cite{commDdef1}, the BSP model~\cite{BSPintro}, or some custom models~\cite{SPD}. We analyze replication in the BSP model, since it is the most widely used among these for practical works. It also accounts for communicated data volume (unlike the communication delay model), and it offers a convenient schedule format for our algorithms. Even when not named explicitly, scheduling algorithms are often designed and analyzed in BSP-like models in several applications, such as sparse triangular system solving~\cite{zarebavani2022hdagg,sptrsv_opdas}.

\subsubsection*{Without replication}

In the BSP model, a schedule is split into larger batches of computations called \textit{supersteps}. Each superstep consists of two phases: first a computation phase where each processor executes operations, but no data movement happens, and then a communication phase where data is exchanged between processors, but no computation happens.

Let $S$ denote the number of supersteps we choose to have in our schedule. Formally, a BSP schedule for a DAG consists of:
\begin{itemize}[leftmargin=1.8em, topsep=2pt, itemsep=1.5pt]
 \item computation phases $V_{p,s} \! \subseteq \! V$ for all $p \! \in \! [P]$, $s \! \in \! [S]$, containing the set of nodes executed by processor $p$ in superstep $s$; 
 \item communication phases $\Gamma_{p,s} \! \subseteq \! V \! \times \! [P]$ for all $p \! \in \! [P]$ and $s \! \in \! [S]$, i.e.\ the set of values (outputs of specific nodes) to send from processor $p$ to some other processor $p'$ in superstep $s$.
\end{itemize}

Let us say that a node $v$ is \textit{already present} on processor $p$ in superstep $s$ if either $\exists \, s' \! \leq \! s$ such that $v \! \in \! V_{p,s'}$, or if $\exists \, s' \! < \! s$ and $\exists \, p' \! \in \! [P]\setminus \{ p\}$ such that $(v, p) \! \in \! \Gamma_{p',s'}$. A valid BSP schedule must satisfy the following properties for all $p \! \in \! [P]$ and $s \! \in \! [S]$:
\begin{itemize}[leftmargin=1.8em, topsep=2pt, itemsep=1.5pt]
 \item All inputs of the operations are available: for all $v \! \in \! V_{p,s}$, all the parents $u$ of $v$ in $G$ are already present on $p$ in superstep $s$.
 \item Any data we send is available: for all $(v,p') \! \in \! \Gamma_{p,s}$, $v$ is already present on $p$ in superstep $s$.
\end{itemize}
Furthermore, a valid schedule must ensure that all nodes are computed, i.e.\ for all $v \! \in \! V$, there exists a $p \! \in \! [P]$ and $s \! \in \! [S]$ such that $v \! \in \! V_{p,s}$. In scheduling without replication, we also require that each node is assigned only to one processor and superstep, i.e.\ the computation phases $V_{p,s}$ are all disjoint.

The cost of a BSP schedule is defined as follows. For each superstep $s \! \in \! [S]$, the cost of the compute phase is $\max_{p \in [P]} |V_{p,s}|$, or $\max_{p \in [P]} \sum_{v \in V_{p,s}} \omega(v)$ with node weights, since the computations on the different processors run in parallel. The communication cost is more complex, defined via a so-called $h$-relation and two global parameters $g$, $L$ of the BSP model: $g$ is the cost of sending a single unit of data (the inverse bandwidth), and $L$ is a fixed synchronization cost incurred by each superstep. Let us denote the values received by processor $p$ in superstep $s$ by $\Gamma'_{p,s} = \{ (v, p') \! \in \! V \! \times \! [P] \, | \, (v,p) \! \in \! \Gamma_{p',s}  \}$. The cost of a communication phase depends on the maximal amount of data sent or received by any processor, defined as
\[ L + g \cdot \max_{p \in [P]} \, \max \, ( \, |\Gamma_{p,s}| \, ,  \, |\Gamma'_{p,s}| ) \, .\]
In case of node weights for communication, the two terms in the max are replaced by $\sum_{(v,p') \in \Gamma_{p,s}} \mu(v)$ and $\sum_{(v,p') \in \Gamma'_{p,s}} \mu(v)$, respectively. The cost of a superstep is the sum of the compute and communication costs, and the total cost of a schedule is the sum of the costs of all supersteps. Unlike partitioning, this approach combines workload balance and communication costs into a single objective function. We refer to~\cite{BSP_DAG_opdas, BSPbook2} for more details on the BSP model.

Given an input computational DAG and parameters $P$, $g$ and $L$, our goal is to find a valid BSP schedule of minimal cost.

\begin{figure}[t]
    \centering
    \hspace{3pt}
     \resizebox{0.45\textwidth}{!}{\begin{tikzpicture}
	
    \begin{scope}[very thick, dashed, gray]
    \draw (0pt,0pt) rectangle (50pt,30pt);
    \draw (0pt,50pt) rectangle (50pt,80pt);
    \draw (80pt,0pt) rectangle (130pt,30pt);
    \draw (80pt,50pt) rectangle (130pt,80pt);
    \draw (160pt,0pt) rectangle (210pt,30pt);
    \draw (160pt,50pt) rectangle (210pt,80pt);
    \end{scope}

    \node[anchor=center] at (25pt,-8pt) {\small $V_{2,1}$};
    \node[anchor=center] at (105pt,-8pt) {\small $V_{2,2}$};
    \node[anchor=center] at (185pt,-8pt) {\small $V_{2,3}$};
    \node[anchor=center] at (25pt,87pt) {\small $V_{1,1}$};
    \node[anchor=center] at (105pt,87pt) {\small $V_{1,2}$};
    \node[anchor=center] at (185pt,87pt) {\small $V_{1,3}$};

    \begin{scope}[thick, arrows=-stealth]
    \draw (8pt,15pt) -- (21pt,8pt);
    \draw (8pt,15pt) -- (21pt,22pt);
    \draw (25pt,8pt) -- (39pt,13pt);
    \draw (25pt,22pt) -- (39pt,17pt);

    \draw (8pt,65pt) -- (21pt,65pt);
    \draw (25pt,65pt) -- (38pt,65pt);

    \draw (89pt,20pt) -- (102pt,20pt);
    \draw (89pt,19pt) -- (118.5pt,9pt);
    \draw (106pt,20pt) -- (119pt,13pt);

    \draw (95pt,65pt) -- (111pt,65pt);

    \draw (42pt,65pt) -- (91pt,65pt);
    \draw (42pt,65pt) -- (87pt,22pt);

    \draw (42pt,15pt) -- (93pt,62pt);
    \draw (25pt,23pt) -- (86pt,20pt);

    \draw (122pt,10pt) -- (166pt,7pt);
    \draw (170pt,7pt) -- (198pt,7pt);

    \draw (175pt,58pt) -- (191pt,58pt);
    \draw (175pt,72pt) -- (191pt,72pt);

    \draw[very thick, gray] (165pt,29pt) -- (178pt,16.5pt);
    \draw[densely dashed] (95pt,65pt) -- (200pt,11pt);

    \draw (122pt,10pt) -- (173pt,56pt);

    \draw (115pt,65pt) -- (172pt,59pt);
    \draw (115pt,65pt) -- (172pt,71pt);
    
    \end{scope}

    \draw[black, fill=white] (8pt,15pt) circle (0.9ex);
    \draw[black, fill=white] (25pt,8pt) circle (0.9ex);
    \draw[black, fill=white] (25pt,22pt) circle (0.9ex);
    \draw[black, fill=white] (42pt,15pt) circle (0.9ex);

    \draw[black, fill=white] (8pt,65pt) circle (0.9ex);
    \draw[black, fill=white] (25pt,65pt) circle (0.9ex);
    \draw[black, fill=white] (42pt,65pt) circle (0.9ex);

    \draw[black, fill=white] (89pt,20pt) circle (0.9ex);
    \draw[black, fill=white] (106pt,20pt) circle (0.9ex);
    \draw[black, fill=white] (122pt,10pt) circle (0.9ex);

    \draw[black, fill=lightgray] (95pt,65pt) circle (1.0ex);
    \node[anchor=center] at (95pt,65pt) {$v$};
    \draw[black, fill=white] (115pt,65pt) circle (0.9ex);

    \draw[black, fill=white] (170pt,7pt) circle (0.9ex);
    \draw[black, fill=lightgray] (180pt,14pt) circle (1.0ex);
    \node[anchor=center] at (180pt,14pt) {$v$};
    \draw[black, fill=white] (202pt,7pt) circle (0.9ex);

    \draw[black, fill=white] (175pt,58pt) circle (0.9ex);
    \draw[black, fill=white] (175pt,72pt) circle (0.9ex);
    \draw[black, fill=white] (195pt,58pt) circle (0.9ex);
    \draw[black, fill=white] (195pt,72pt) circle (0.9ex);

\end{tikzpicture}}
     \vspace{-6pt}
    \caption{BSP schedule of a DAG on $2$ processors and $3$ supersteps. Node $v$ is originally computed in $V_{1,2}$. Replicating $v$ in $V_{2,3}$ allows to remove the dashed communication. Note: both parents of $v$ are already present on processor $2$ by $V_{2,3}$.}
    \label{fig:bsp}
\end{figure}
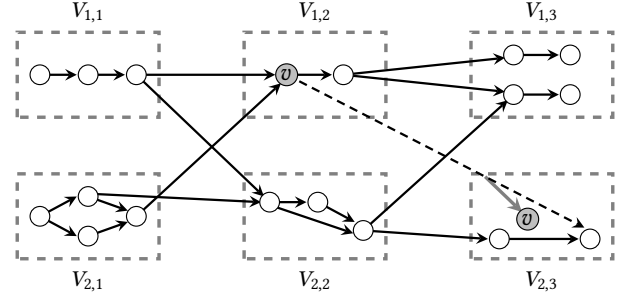

\subsubsection*{With replication}

The BSP scheduling problem is straightforward to extend with replication: we simply remove the restriction that the $V_{p,s}$ pairs need to be disjoint. The rest of the problem definition and the cost function remains identical.

Note that even with replication, it is never beneficial to assign a node twice to the same processor in different supersteps. However, replicating a node on two different processors can indeed help: it might increase the cost of the corresponding computation phases, but it can also notably decrease the incurred communication costs, thus altogether providing a schedule with lower cost.

An example schedule with replication is illustrated in Figure~\ref{fig:bsp}.

\section{Theoretical results}

We first briefly analyze how introducing replication affects the partitioning and scheduling problems. It is not surprising that allowing replication can significantly decrease the optimal cost. Specifically, in partitioning, it can easily happen that the optimal cost without replication is $\Omega(n^2)$, but with replication it is $0$: a simple example graph for $P\!=\!2$ consists of two cliques of size $\frac{1+\varepsilon}{2} \! \cdot \! n$ that intersect in $\varepsilon \! \cdot \! n$ nodes. In scheduling, replication can reduce the cost by a factor $(P-\delta)$ for any $\delta>0$; this is essentially as high as possible, since the optimal cost is always between $\frac{n}{P}$ and $n$. The example here is a complete bipartite DAG with a set $U$ of source nodes and sets $U_1$, ..., $U_P$ of sink nodes, and carefully chosen parameters. Without replication, for a large constant $g$, the optimum only uses one processor at a cost of $n$. However, if we can replicate $U$, then the sets $U_i$ can be parallelized, and the optimum is essentially $\frac{n}{P}$ if $|U| \!<\!< \!|U_i|$. We discuss both of these examples in detail in Appendix~\ref{app:reduction}.

Below we focus on a more challenging question: the impact of replication on the complexity of the problems. Since it is usually straightforward to show NP-hardness for both partitioning and scheduling (in many model variants, with/without replication), we instead focus on the approximability of the optimal solution.

\subsection{Graph Partitioning}

For partitioning, we show that the optimal solution is NP-hard to approximate to any finite factor. This already holds for graph partitioning, and hence naturally carries over to hypergraphs, too.

This provides an interesting contrast to partitioning without replication. Without replication, for simple graphs, the optimum is known to be approximable in polynomial time to a $O(\log n)$ factor~\cite{Feige00, Krauthgamer06, Racke08}. For hypergraphs, the landscape is more complex. The best known inapproximability results depend on the complexity-theoretic conjectures we accept~\cite{hypergraphs_opdas}: $n^{1/\text{poly} \log \log n}$, assuming the Exponential Time Hypothesis~\cite{ETHhardness}; $n^{o(1)}$, assuming Gap-ETH~\cite{ETHhardness}; $n^{\delta}$ for some $\delta > 0$, assuming specific one-way functions~\cite{crypt}; and $n^{1/12-\delta}$ assuming the Hypergraph Dense vs.\ Random Conjecture~\cite{DenseVsSparseC}. However, all of these are significantly weaker than the inapproximability bound in our claim, which only relies on $P \! \neq \! NP$. For the case of hypergraph partitioning with $\varepsilon\!=\!0$, the optimum is also known to be approximable to a $\tilde{O}(\sqrt{n})$-factor~\cite{BisectionApprox}.

As such, our result indicates that introducing replication significantly increases the complexity of the partitioning problem.

\begin{theorem} \label{th:part_inapprox}
For any $\varepsilon \! \in \! (0,1)$, graph partitioning with replication is NP-hard to approximate to any finite factor.
\end{theorem}

\renewcommand*{\proofname}{Proof sketch}

\begin{proof}
The proof relies on showing that it is already NP-hard to decide whether there exists a partitioning with replication of cost $0$. Hence it is already NP-hard to distinguish between an optimum cost of $0$ or $1$; the theorem already follows from this.

The main idea of the proof is quite natural: if a partitioning has cost $0$, i.e.\ there are no edges between $V_1 \setminus V_2$ and $V_2 \setminus V_1$ at all, then $V_1 \cap V_2$ is a so-called vertex separator in the graph. As such, the problem of partitioning with cost $0$ allows a simple reduction from the balanced vertex separator problem, where the allowed size of the separator is set to $\varepsilon \! \cdot \! n$, that is, the maximal allowed size of the intersection $V_1 \cap V_2$ in the partitioning problem.

The concrete proof is somewhat more complex because previous works on balanced vertex separator typically consider the separator size as a problem input. Hence from these works, it only follows that balanced vertex separator is NP-hard for \textit{some} separator sizes. To prove our theorem for any $\varepsilon \! \in \! (0, 1)$, we instead require that it is NP-hard for a separator size of $\varepsilon \cdot n$ for any $\varepsilon \! \in \! (0, 1)$. As such, our proof needs some further technical steps to show that balanced vertex separator is NP-hard with separator size $\varepsilon \cdot n$ for any $\varepsilon \! \in \! (0, 1)$.
\end{proof}

The proof above relies on differentiating between an optimum of $0$ and $1$, but we can also extend this to optimum values of $0$ and $n^{2-\delta}$. This shows inapproximability to a large additive term. 

\begin{theorem} \label{th:additive}
For any $\varepsilon \! \in \! (0,1)$ and $\delta \!> \!0$, graph partitioning with replication is NP-hard to approximate to an $n^{2-\delta}$ additive term.
\end{theorem}

\begin{proof}
As a reduction from the previous theorem, for any graph $G$ on $N$ nodes, we create a new graph $G'$ on $n=N^{1/\delta}$ nodes by replacing each node in $G$ by a clique of size $n^{1-\delta}$. The construction remains polynomial for any fixed $\delta$. Original edges are replaced by complete bipartite connections between the cliques. We show that this expanded graph $G'$ fulfills two properties:
\begin{itemize}[leftmargin=1.8em, topsep=2pt, itemsep=1.5pt]
 \item The optimum in $G$ is 0 if and only if the optimum in $G'$ is $0$.
 \item If the optimum in $G'$ is not $0$, then it is at least $n^{2-\delta}$.
\end{itemize}
This shows that it is NP-hard to differentiate between an optimum of $0$ and $n^{2-\delta}$, hence the theorem follows. The technical part is to prove the second point above: we use several steps to transform any partitioning into an easier-to-handle solution, and then analyze several cases for the lower bound. We also need some additions to the construction to handle the inconvenient fact that we may have $\lfloor \frac{1+\varepsilon}{2} \! \cdot \! N \! \cdot \! n^{1-\delta} \rfloor \neq \lfloor \frac{1+\varepsilon}{2} \! \cdot \! N \rfloor \! \cdot \! n^{1-\delta}$, so the clique-expansion of an original partition may not fit into the constraint bounds anymore.
\end{proof}

\subsection{DAG scheduling}

Complexity questions are often more challenging to study for scheduling problems where computation and communication costs are combined into the same objective. For the simplest case of free communication, it is not even known whether scheduling for a constant $P \geq 3$ is NP-hard. For more realistic models, NP-hardness has been long known, and there are also results on the classes of DAGs for which scheduling becomes NP-hard. For BSP in particular, it is known that scheduling is polynomially solvable in chain DAGs (where all in- and outdegrees are at most $1$), but already NP-hard in in-trees (where out-degrees are at most $1$) and $2$-layer DAGs (where any directed path has length at most $1$)~\cite{BSP_DAG_opdas}. One can observe that these properties remains unchanged even with replication.

\begin{lemma} \label{lem:nphard}
For $P \in O(1)$, BSP scheduling with replication is solvable in polynomial time on chain DAGs, but NP-hard on in-trees and $2$-layer DAGs.
\end{lemma}

The proof is deferred to Appendix~\ref{app:nphard}. On a high level, for chain DAGs and in-trees, the out-degree of at most $1$ ensures that replication is never beneficial, hence the original proofs carry over to our setting. For the case of $2$-layer DAGs, some further observations are also required to adapt the original construction~\cite{BSP_DAG_opdas}.

The approximability of the optimal schedule is a more challenging question. For $P \! \in \! O(1)$, the only hardness result we are aware of is a $(1+\delta)$-factor inapproximability in the BSP model, for a small $\delta>0$~\cite{BSP_DAG_opdas}. Intuitively, proving inapproximability is difficult because any schedule has an inherent compute cost of at least $\frac{n}{P}$, and the cost difference between solutions is often small compared to this. To overcome a similar issue in pebbling problems, recent work considered an alternative cost metric which directly subtracts this unavoidable cost from the total~\cite{MPP_opdas}. This idea generalizes well to scheduling with replication, since it naturally captures any extra cost incurred by either communication steps or replications.

\begin{definition}
The \emph{surplus cost} of a BSP schedule (with/without replication) is obtained by subtracting $\frac{n}{P}$ from its original BSP cost.
\end{definition}

To our knowledge, surplus cost has not been studied before in BSP or other models. We show that the optimal surplus cost cannot be approximated to a reasonable factor, either with or without replication. This result is a bit of an outlier in the paper, as it does not highlight a difference between the replicating and non-replicating case. Nonetheless, it provides valuable insight into DAG scheduling problems, and we hope that it motivates the study of surplus-like metrics in the future.

\begin{theorem} \label{thm:surplus}
Assuming the exponential time hypothesis, for $P \! = \! 2$ and some $\delta>0$, the surplus cost in BSP scheduling (with or without replication) is NP-hard to approximate to an $n^{1/(\log\log n)^{\delta}}$ factor.
\end{theorem}

\begin{proof}
The main idea is similar to a $2$-layer DAG construction that is used to show NP-hardness in~\cite{BSP_DAG_opdas}. By changing some gadgets and the analysis, we obtain a reduction from the smallest $k$-edge subgraph problem, which is inapproximable to a $n^{1/(\log\log n)^{\delta}}$ factor under the same assumption~\cite{ETHhardness}.

On a high level, the construction contains a block gadget for each edge of the original graph (the input of smallest $k$-edge subgraph). With extra gadgets and careful adjustments, we can ensure that any reasonable schedule assigns exactly $k$ such gadgets to processor $1$, and the rest to processor $2$. We then create a gadget for each node of the original graph, which incurs a communication or replication step if and only if one of the incident edge gadgets is assigned to processor $1$. The rest of the construction allows for perfect work balance and requires no communication, so the surplus cost corresponds exactly to the communication or replication steps incurred by the node gadgets that are incident to one of the $k$ chosen edges. This allows a reduction to smallest $k$-edge subgraph. The cost of communication/replication is actually split between the two processors, so it might be a factor $2$ off from the original cost, but this has no affect on approximability to a superconstant factor. The proof details are in Appendix~\ref{app:spes}.
\end{proof}

\section{Replication in partitioning} \label{sec:part}

We now turn to concrete algorithmic approaches to study the impact of replication empirically. For hypergraph partitioning, we consider a standard representation of the problem as an ILP, and we analyze two different ways to extend this with replication. 

In the formulations, $v \! \in \! V$, $e \! \in \!E$ and $p \! \in \! [P]$ will always denote a node, hyperedge and processor, respectively. For simplicity, we write $\sum_v$, $\sum_e$ and $\sum_p$ instead of $\sum_{v \in V}$, $\sum_{e \in E}$ and $\sum_{p \in [P]}$. 

\subsection{ILP formulation of hypergraph partitioning}

Hypergraph partitioning is known to allow a rather straightforward ILP formulation~\cite{jenneskens22}; we briefly outline it here for completeness.

For each node $v$ and processor $p$, one can use a binary variable $x_{v,p}$ to indicate whether $v \! \in \! V_p$. For each hyperedge $e$, a binary variable $y_{e,p}$ indicates whether there exists any $v \! \in \! e$ such that $v \! \in \! V_p$. We use $\sum_{p} x_{v,p} = 1$ for all $v$ to ensure that each vertex is assigned. For all hyperedges $e$, all $v \! \in \! e$ and all processors $p$, the constraint $y_{e,p} \geq x_{v,p}$ ensures that the variables $y_{e,p}$ fulfill their role. The balance constraint can be expressed as $\sum_v x_{v,p} \leq \frac{1+\varepsilon}{P} \cdot n$. As the objective function, we aim to minimize $\sum_e \, (\sum_p y_{e,p} - 1)$, which describes the $(\lambda-1)$-metric.

\subsection{Introducing replication}

We discuss two ways to extend this ILP formulation with replication.

\subsubsection{ILP/D -- Duplication} \label{sec:ilp_dupl}

We first consider a limited form of replication: we only allow any specific node to be replicated on at most $2$ processors. This is a restriction from the modeling side, but it allows for a simpler ILP representation.

We begin by introducing a binary variable $z_v$ for each node $v$ to indicate whether $v$ is replicated on two processors. We can replace the original constraint $\sum_p x_{v,p} = 1$ by $\sum_p x_{v,p} = z_v + 1$; this implicitly ensures that $\sum_p x_{v,p} \in \{ 1, 2\}$.

The variables $y_{e,p}$ now indicate whether the solution uses processor $p$ to cover some part of hyperedge $e$. E.g.\ if $e=\{ v_1, v_2\}$, with $v_1 \in V_1 \cap V_2$ and $v_2 \in V_1 \setminus V_2$, the desired values are $y_{e,1}=1$ but $y_{e,2}=0$, despite the fact that $e$ intersects $V_2$. 
We ensure this with the following constraints. If a vertex is not duplicated, i.e.\ $z_v = 0$, then we need $y_{e,p} \geq x_{v,p}$ as before. We can ensure this with the constraint $y_{e,p} \geq x_{v,p}-z_v$, which has no effect if $z_v = 1$. If a vertex is duplicated, i.e., $z_v = 1$, then for each hyperedge $e$ with $v \in e$ and each pair of processors $p_1$ and $p_2$, we need to ensure that $x_{v,p_1}=x_{v,p_2}=1$ implies either $y_{e,p_1}=1$ or $y_{e,p_2}=1$. This can be formulated as $y_{e,p_1} + y_{e,p_2} \geq x_{v,p_1} + x_{v, p_2} -1$.

This results in $P^2 \! \cdot \! \kappa$ constraints of this type (recall that $\kappa$ is the number of pins), which is a notable increase from the original $P \! \cdot \! \kappa$. However, the magnitude of the number of variables remains unchanged, so for a small $P$, this may still be viable.

\subsubsection{ILP/R -- Unlimited replication} \label{sec:ilp_repl}

An alternative approach is to introduce a binary variable $z_{(v,e),p}$ for each pin $(v,e)$ in the hypergraph, and each processor $p$. This indicates whether the replica of node $v$ on processor $p$ is used to cover $v$ in hyperedge $e$.

The constraint $\sum_{p} z_{(v,e),p} = 1$ for all pins $(v,e)$ ensures that each pin is covered in each hyperedge. For each pin $(v,e)$ and processor $p$, the constraint $x_{v,p} \geq z_{(v,e),p}$ ensures that a node is only used if it is indeed assigned to a processor. Finally, for each pin $(v,e)$ and processor $p$, we add $y_{e,p} \geq z_{(v,e),p}$: if the pin is covered by $p$, then $p$ is one of the processors used for hyperedge $e$. We can again remove the original constraint $\sum_{p} x_{v,p} = 1$.

This provides a more general ILP, but the number of variables now also scales with the number of pins $\kappa$. As such, this ILP has more variables both in the primal and dual problem, possibly making it more challenging for an ILP solver.

\section{Replication in scheduling}

Intuitively, the scheduling problem has a higher degree of freedom than partitioning, and hence it is in many ways more challenging to tackle. Due to this, algorithm design for scheduling often focuses on sophisticated heuristics instead of exact solutions.

\subsection{Scheduling baselines}

As a baseline, we consider the recent work of Papp et al.~\cite{BSP_algos_opdas}, which develops a framework to optimize BSP schedules for DAGs. In particular, we use the following two-step algorithm from~\cite{BSP_algos_opdas}:
\begin{itemize}[leftmargin=1.8em, topsep=2pt, itemsep=1.5pt]
\item first a good initial schedule is developed using a well-balanced heuristic (named BSPg), which adapts the classical list scheduling approach to the BSP model;
\item this is followed by a hill climbing local search, which further improves the schedule by executing small modifications that decrease the cost.
\end{itemize}
The framework of~\cite{BSP_algos_opdas} also considers more advanced methods to further improve small parts of the schedule, but these provide very limited gain at the cost of a very high running time.

We point out that our goal is not to compete with this baseline scheduler, but to take the non-replicating BSP schedule it produces, introduce replication steps into this schedule, and analyze how much this can decrease the cost of the baseline schedule.

\subsection{Introducing replication}

We consider two ways to introduce replication into a BSP schedule.

\subsubsection{ILP-based method}

One natural idea is to once again embed replication into an ILP formulation of the scheduling problem. The ILP formulation of BSP scheduling is rather complex, but it has already been analyzed~\cite{BSP_DAG_opdas, BSP_algos_opdas}. Replication only needs a trivial change in this formulation: relaxing the constraint that all nodes must be computed exactly once.

However, solving scheduling problems via ILPs is much more challenging than in case of partitioning: the scheduling ILP has a much higher number of variables due to the extra dimension of time steps. Even with modern ILP solvers, this approach is only viable for very small graphs in practice.

\subsubsection{The `basic' heuristic}

We also consider some heuristic method to also tackle the scheduling problem on much larger DAGs. These algorithms begin from a BSP schedule without replication, and in a local-search-like fashion, they identify improvement opportunities where replacing some communication step(s) in the schedule by replication decreases the total cost.

The basic heuristic only considers the simplest improvement steps: replacing a single communication step with replication. That is, for any communication step of sending node $v$ from processor $p_1$ to $p_2$, we check whether the total cost decreases if we remove this step, and instead replicate $v$ on processor $p_2$ in an appropriate superstep. The first possible superstep to replicate $v$ on $p_2$ is the first computation phase where all of $v$'s parents are already present on $p_2$; the last possible superstep is where $v$ is first required for either a computation or communication on $p_2$. We select the superstep from this interval where adding the computation of $v$ on $p_2$ increases the compute cost by the smallest amount. If this cost increase is smaller than the gain from removing the communication step, then we execute this modification, decreasing the total BSP cost.

Note that executing such a step might also enable further improvement steps that were not beneficial before: the gain from removing other communications may become larger, or the cost increase from an extra computation may become smaller. The basic heuristic simply keeps searching for and executing these improvement steps as long as it is possible.

\subsubsection{The `advanced' heuristic}

We also implement a more advanced heuristic that considers several ways to execute larger modifications, replacing multiple communication steps by replication in a single step. This allows to further improve the schedule even when the basic heuristic is stuck in a local minimum. Our algorithm uses three different approaches to further utilize replication in the schedule (also shown in Figure~\ref{fig:heuristic}).

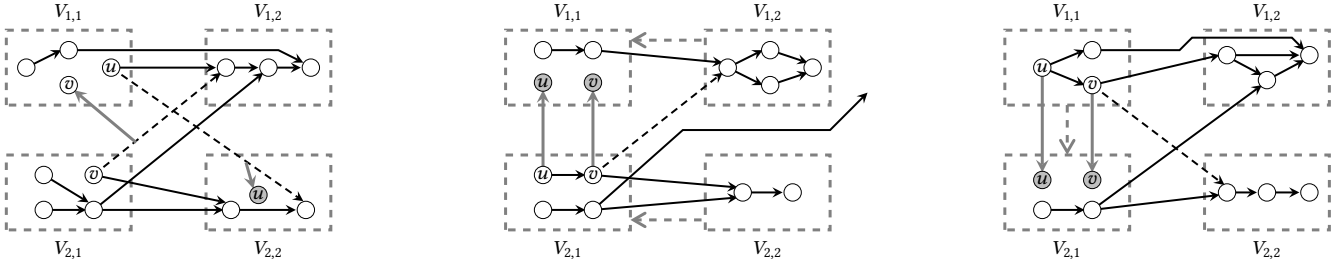
\begin{figure*}[t]
    \centering
    \hspace{2pt}
     \resizebox{0.99\textwidth}{!}{\begin{tikzpicture}
	
    \begin{scope}[very thick, dashed, gray]
    \draw (0pt,0pt) rectangle (50pt,30pt);
    \draw (0pt,50pt) rectangle (50pt,80pt);
    \draw (80pt,0pt) rectangle (130pt,30pt);
    \draw (80pt,50pt) rectangle (130pt,80pt);
    \end{scope}

    \node[anchor=center] at (25pt,-8pt) {\small $V_{2,1}$};
    \node[anchor=center] at (105pt,-8pt) {\small $V_{2,2}$};
    \node[anchor=center] at (25pt,87pt) {\small $V_{1,1}$};
    \node[anchor=center] at (105pt,87pt) {\small $V_{1,2}$};

    \begin{scope}[thick, arrows=-stealth]
    \draw (15pt,8pt) -- (32pt,8pt);
    \draw (15pt,22pt) -- (33pt,11pt);
    \draw (35pt,8pt) -- (87pt,8pt);
    \draw (35pt,22pt) -- (88pt,11pt);

    \draw (8pt,65pt) -- (22pt,72pt);
    \draw (42pt,65pt) -- (84pt,65pt);
    \draw[densely dashed] (42pt,65pt) -- (119pt,11.5pt);
    \draw (25pt,72pt) -- (109pt,72pt) -- (119pt,67pt);

    \draw (90pt,8pt) -- (116pt,8pt);

    \draw (88pt,65pt) -- (101pt,65pt);
    \draw (105pt,65pt) -- (118pt,65pt);
    \draw[densely dashed] (35pt,22pt) -- (86pt,63pt);
    \draw (35pt,8pt) -- (103pt,63pt);

    \draw[very thick, gray] (52pt,35pt) -- (27pt,56pt);
    \draw[very thick, gray] (96pt,27pt) -- (99pt,16.5pt);
    
    \end{scope}

    \draw[black, fill=white] (15pt,22pt) circle (0.9ex);
    \draw[black, fill=white] (15pt,8pt) circle (0.9ex);
    \draw[black, fill=white] (35pt,22pt) circle (0.9ex);
    \draw[black, fill=white] (35pt,8pt) circle (0.9ex);
    \node[anchor=center] at (35pt,22pt) {\small $v$};

    \draw[black, fill=white] (8pt,65pt) circle (0.9ex);
    \draw[black, fill=white] (25pt,72pt) circle (0.9ex);
    \draw[black, fill=white] (42pt,65pt) circle (0.9ex);
    \node[anchor=center] at (42pt,65pt) {\small $u$};
    \draw[black, fill=white] (25pt,58pt) circle (0.9ex);
    \node[anchor=center] at (25pt,58pt) {\small $v$};

    \draw[black, fill=white] (90pt,8pt) circle (0.9ex);
    \draw[black, fill=white] (120pt,8pt) circle (0.9ex);
    \draw[black, fill=lightgray] (101pt,14pt) circle (0.9ex);
    \node[anchor=center] at (101pt,14pt) {\small $u$};

    \draw[black, fill=white] (88pt,65pt) circle (0.9ex);
    \draw[black, fill=white] (105pt,65pt) circle (0.9ex);
    \draw[black, fill=white] (122pt,65pt) circle (0.9ex);


\begin{scope}[very thick, dashed, gray]
    \draw (200pt,0pt) rectangle (250pt,30pt);
    \draw (200pt,50pt) rectangle (250pt,80pt);
    \draw (280pt,0pt) rectangle (330pt,30pt);
    \draw (280pt,50pt) rectangle (330pt,80pt);
    \draw[arrows=angle 60-] (250pt,76pt) -- (280pt,76pt);
    \draw[arrows=angle 60-] (250pt,4pt) -- (280pt,4pt);
    \end{scope}

    \node[anchor=center] at (225pt,-8pt) {\small $V_{2,1}$};
    \node[anchor=center] at (305pt,-8pt) {\small $V_{2,2}$};
    \node[anchor=center] at (225pt,87pt) {\small $V_{1,1}$};
    \node[anchor=center] at (305pt,87pt) {\small $V_{1,2}$};

    \begin{scope}[thick, arrows=-stealth]
    \draw (215pt,8pt) -- (232pt,8pt);
    \draw (215pt,22pt) -- (232pt,22pt);
    \draw (235pt,8pt) -- (293pt,13pt);
    \draw (235pt,22pt) -- (293pt,17pt);
    \draw[densely dashed] (235pt,22pt) -- (287pt,63pt);
    \draw (235pt,8pt) -- (271pt,40pt) -- (330pt,40pt) -- (345pt,55pt);

    \draw (215pt,72pt) -- (232pt,72pt);
    \draw (235pt,72pt) -- (287pt,67pt);

    \draw (295pt,15pt) -- (311pt,15pt);

    \draw (289pt,65pt) -- (303pt,72pt);
    \draw (289pt,65pt) -- (303pt,58pt);
    \draw (306pt,72pt) -- (321pt,67pt);
    \draw (306pt,58pt) -- (321pt,63pt);

    \draw[very thick, gray] (215pt,22pt) -- (215pt,56pt);
    \draw[very thick, gray] (235pt,22pt) -- (235pt,56pt);
    
    \end{scope}

    \draw[black, fill=white] (215pt,22pt) circle (0.9ex);
    \node[anchor=center] at (215pt,22pt) {\small $u$};
    \draw[black, fill=white] (215pt,8pt) circle (0.9ex);
    \draw[black, fill=white] (235pt,22pt) circle (0.9ex);
    \node[anchor=center] at (235pt,22pt) {\small $v$};
    \draw[black, fill=white] (235pt,8pt) circle (0.9ex);

    \draw[black, fill=lightgray] (215pt,59pt) circle (0.9ex);
    \node[anchor=center] at (215pt,59pt) {\small $u$};
    \draw[black, fill=white] (215pt,72pt) circle (0.9ex);
    \draw[black, fill=lightgray] (235pt,59pt) circle (0.9ex);
    \node[anchor=center] at (235pt,59pt) {\small $v$};
    \draw[black, fill=white] (235pt,72pt) circle (0.9ex);

    \draw[black, fill=white] (295pt,15pt) circle (0.9ex);
    \draw[black, fill=white] (315pt,15pt) circle (0.9ex);

    \draw[black, fill=white] (289pt,65pt) circle (0.9ex);
    \draw[black, fill=white] (306pt,58pt) circle (0.9ex);
    \draw[black, fill=white] (306pt,72pt) circle (0.9ex);
    \draw[black, fill=white] (323pt,65pt) circle (0.9ex);


\begin{scope}[very thick, dashed, gray]
    \draw (400pt,0pt) rectangle (450pt,30pt);
    \draw (400pt,50pt) rectangle (450pt,80pt);
    \draw (480pt,0pt) rectangle (530pt,30pt);
    \draw (480pt,50pt) rectangle (530pt,80pt);
    \draw[arrows=-angle 60] (425pt,50pt) -- (425pt,30pt);
    \end{scope}

    \node[anchor=center] at (425pt,-8pt) {\small $V_{2,1}$};
    \node[anchor=center] at (505pt,-8pt) {\small $V_{2,2}$};
    \node[anchor=center] at (425pt,87pt) {\small $V_{1,1}$};
    \node[anchor=center] at (505pt,87pt) {\small $V_{1,2}$};

    \begin{scope}[thick, arrows=-stealth]
    \draw (415pt,8pt) -- (432pt,8pt);
    \draw (435pt,8pt) -- (486pt,14pt);
    \draw (435pt,8pt) -- (503pt,58pt);

    \draw (415pt,65pt) -- (432pt,72pt);
    \draw (415pt,65pt) -- (432pt,58pt);
    \draw (435pt,72pt) -- (470pt,72pt) -- (475pt,77pt) -- (515pt,77pt) -- (520pt,72pt);
    \draw (435pt,58pt) -- (486pt,70pt);
    \draw[densely dashed] (435pt,58pt) -- (488pt,18pt);

    \draw (488pt,15pt) -- (501pt,15pt);
    \draw (505pt,15pt) -- (518pt,15pt);

    \draw (488pt,70pt) -- (518pt,70pt);
    \draw (488pt,70pt) -- (503pt,62pt);
    \draw (505pt,60pt) -- (520pt,68pt);

    \draw[very thick, gray] (415pt,65pt) -- (415pt,23pt);
    \draw[very thick, gray] (435pt,58pt) -- (435pt,23pt);
    
    \end{scope}

    \draw[black, fill=lightgray] (415pt,20pt) circle (0.9ex);
    \node[anchor=center] at (415pt,20pt) {\small $u$};
    \draw[black, fill=white] (415pt,8pt) circle (0.9ex);
    \draw[black, fill=lightgray] (435pt,20pt) circle (0.9ex);
    \node[anchor=center] at (435pt,20pt) {\small $v$};
    \draw[black, fill=white] (435pt,8pt) circle (0.9ex);

    \draw[black, fill=white] (415pt,65pt) circle (0.9ex);
    \node[anchor=center] at (415pt,65pt) {\small $u$};
    \draw[black, fill=white] (435pt,72pt) circle (0.9ex);
    \draw[black, fill=white] (435pt,58pt) circle (0.9ex);
    \node[anchor=center] at (435pt,58pt) {\small $v$};

    \draw[black, fill=white] (489pt,15pt) circle (0.9ex);
    \draw[black, fill=white] (505pt,15pt) circle (0.9ex);
    \draw[black, fill=white] (522pt,15pt) circle (0.9ex);

    \draw[black, fill=white] (489pt,70pt) circle (0.9ex);
    \draw[black, fill=white] (505pt,60pt) circle (0.9ex);
    \draw[black, fill=white] (522pt,70pt) circle (0.9ex);

\end{tikzpicture}}
     \vspace{-10pt}
    \caption{Examples for each ingredient of the advanced heuristic. In batch replication (left), we remove the $2$ dashed communications simultaneously, replicating $u$ and $v$. In superstep merging (middle), we merge $V_{i,2}$ into $V_{i,1}$ for $i \! \in \! \{ 1,2\}$, and replicate $u$ and $v$ to replace the dashed communication; the other communication can be kept in the merged superstep. In superstep replication (right), we replicate in $V_{2,1}$ all nodes of $V_{1,1}$ that are ever required on processor $2$; this removes the dashed communication.}
    \label{fig:heuristic}
\end{figure*}

The first method is \textit{batch replication}, which replicates single nodes like the basic heuristic, but executes multiple steps at the same time to escape local minima. In particular, the $h$-relation cost in BSP means that if the maximal communication cost is attained in the send and/or receive cost of multiple processors, then all these send and receive costs need to be decreased simultaneously to obtain any decrease in the cost. In batch replication, we consider each superstep in order, and attempt to remove multiple communication steps at the same time; specifically, at least one outgoing/incoming communication from every processor where the send/receive cost is equal to the maximum. We replicate these nodes on the target processors instead, always choosing the superstep (from the interval of valid options) where this increases the compute costs minimally. Once a step is removed from all send/receive costs that saturate the maximum, the communication cost decreases; if the decrease is larger than the total increase in compute costs, we accept the modification. We continue this process as long as it allows to decrease the total cost; otherwise, we proceed to the next superstep.

The second method is \textit{superstep merging}, which considers all consecutive superstep pairs $(s$, $s\!+\!1)$, and attempts to merge them into a single superstep. This requires changes for all communication steps in superstep $s$. Assume that a node $v$ is sent from processor $p_1$ to $p_2$. If $v$ is not yet needed on $p_2$ in superstep $(s\!+\!1)$, we can simply keep this communication step in the new (merged) superstep. If $v$ was already present on $p_1$ before superstep $s$, we can move the communication to the superstep before $s$. Otherwise, $v$ needs to be replicated on $p_2$ in the merged superstep, and we need to make the parents of $v$ also available on $p_2$ by superstep $s$ if not present already. For parents that are already computed before superstep $s$, we can simply send them in the superstep before $s$; however, if a parent of $v$ is first computed on $p_1$ in superstep $s$, we also need to replicate it on $p_2$, and continue the process with the parent's parents in a recursive fashion. Once all the necessary changes are identified, we check the cost of the modified schedule, and execute the merging step if it decreases the cost altogether. Superstep merging is especially useful when the synchronization cost $L$ is large; then it is indeed better to have fewer supersteps, but this is often hard to achieve via single-node improvement steps.

The third method is \textit{superstep replication}, where we take a specific superstep $s$ and pair of processors $p_1$, $p_2$, and consider replicating all nodes of $V_{p_1, s}$ in $V_{p_2, s}$. Intuitively, this may identify a less parallelizable segment of the schedule that was assigned to a single processor, and allows to save communication afterwards by replicating this segment on other processors. During this step, we replicate all nodes in $v \! \in \! V_{p_1, s}$ on $V_{p_2, s}$, except those that are already present on $p_2$ or not needed on $p_2$ at any later point. If $v$ has parents that are not yet present on $p_2$ by superstep $s$, we add communication steps to send these to $p_2$ in superstep $(s\!-\!1)$. As in superstep merging, we analyze the cost change for each such transformation, and only execute those that decrease the total cost of the schedule.

Our advanced heuristic keeps iterating through these 3 methods until none of them allows to further improve the schedule.

\section{Experiments}

\subsection{Experimental setup}

In our experiments, we use a diverse dataset to evaluate the algorithms above. For hypergraph partitioning, we assemble a benchmark from the two applications mentioned before. For MoE partitioning, we use open-source profiled data of expert usage from~\cite{moe_data}, obtained from running two modern LLMs (FP8 quantized Qwen3-235B and 4 bit quantized Deepseek R1) on the Massive Multitask Language Understanding benchmark with various questions on 57 subjects. This describes the $8$-tuple of experts invoked by more than $1.5$ million tokens on each layer. We create hyperedges from the most frequently appearing $8$-tuples on a layer, with the hyperedge weight $\mu(e)$ being the frequency normalized to $[1,10]$. We create a separate hyperedge from each of the first $5$ layers of Qwen and DeepSeek, choosing the frequency threshold such that the number of pins in the given hypergraphs is at least $1\,000$. Any experts that do not appear in these most frequent $8$-tuples are discarded. This gives us $5$ hypergraphs for both Qwen and DeepSeek, with $81\,$--$\,109$ and $220\,$--$\,241$ nodes, respectively. This new dataset, named \texttt{moe-8}, is in itself a small standalone contribution of our paper, which might be a useful benchmark for future works on partitioning.

We also create another version of this dataset where instead of the $8$-tuples, we consider all the ${8 \choose 2}$ pairs of nodes appearing in each $8$-tuple. This may allow deeper insight into the distribution of the most frequently invoked experts; indeed, many approaches for MoE partitioning focus on pairs of experts~\cite{moetuner}. We name this dataset \texttt{moe-2}. As before, we add the most frequently appearing pairs as hyperedges until the number of pins is at least $1\,000$. From the first $5$ layers, we obtain $5$ instances on $80\,$--$\,97$ (Qwen) and $167\,$--$\,239$ (DeepSeek) nodes here. Note that these are actually simple graphs, since we always form hyperedges of size exactly $2$.

For SpMV, we consider a selection of application matrices from the SuiteSparse sparse matrix collection~\cite{davis2011university}. This is similar to the benchmark used in~\cite{jenneskens22}, but with larger matrices. We create two datasets via two different hypergraph models for SpMV, each consisting of $10$ hypergraphs between $92\,$--$\,494$ nodes and $184\,$--$\,1910$ pins. In the fine-grained model~\cite{knigge20, jenneskens22}, each node represents a non-zero, and each hyperedge is a row or column of the matrix, as discussed in Section~\ref{sec:def_part}. In the more coarse-grained row-net model~\cite{catalyurek2002hypergraph}, each node represents a column, and each hyperedge represents a row, connecting the columns with a non-zero in the given row. The datasets are named \texttt{spmv-fg} and \texttt{spmv-rn}, respectively. 

We note that there are also established benchmarks of much larger hypergraphs to evaluate large-scale partitioning heuristics~\cite{hMetis, KaHyPar}; however, these are not viable for ILPs due to their size.

For scheduling, we consider $3$ different datasets of computational DAGs that have previously been used in scheduling-related works:
\begin{itemize}[leftmargin=1.8em, topsep=2pt, itemsep=1.5pt]
 \item \texttt{hdb}: the HyperDAG DataBase introduced in~\cite{BSP_algos_opdas} represents a variety of computations from areas like linear algebra, graph algorithms or machine learning. We use the $3$ largest DAG sets here, giving $52$ DAGs with $1_{\,}000$--$100_{\,}000$ nodes.
 \item \texttt{psdd}: Probabilistic Sentential Decision Diagrams (PSDDs)~\cite{kisa2014psdd} are arithmetic circuits for tractable probabilistic inference, whose evaluation naturally forms an irregular DAG~\cite{shah2022dpu, maene2025klay}. We use PSDD circuits from 
the StarAI Circuit Model Zoo. Excluding tiny DAGs, this gives $29$ DAGs on $2_{\,}500$--$175_{\,}000$ nodes.
 \item \texttt{sptrsv}: the parallel solving of sparse triangular systems can be expressed as a scheduling problem on the DAG representation of a triangular matrix~\cite{zarebavani2022hdagg}. We collect $40$ application matrices from the SuiteSparse matrix collection~\cite{davis2011university}; these correspond to $40$ DAGs with $1_{\,}000$--$150_{\,}000$ nodes.
\end{itemize}
Many of the computations in these benchmarks have a highly irregular structure and are in general hard to parallelize.

We also use the smallest DAG sets from~\cite{BSP_algos_opdas} for the scheduling ILP experiments; this contains $16$ DAGs on $40$--$80$ nodes.

To measure the impact of replication, we consider the ratio between the replicating and non-replicating cost (`cost reduction ratio') for each instance. For datasets, we take the geometric mean of these ratios, and subtract this from $1$ to obtain a mean cost reduction value as a percentage.

We run the algorithms with various parameters that are in line with previous works on partitioning and scheduling~\cite{jenneskens22,BSP_algos_opdas}. For the ILPs, we use a smaller number of processors to limit the total running time of the experiments. Specifically, for hypergraph partitioning, we focus on parameters $P \! \in \! \{ 2, 4\}$ and $\varepsilon \! \in \! \{0.0125, 0.025, 0.05\}$. For the scheduling heuristic, running time is not critical, so we explore a larger set of parameters: our main experiments have $P=8$, $g=4$ and $L=20$, but we also investigate other cases with $P \! \in \! \{2,4,8,16\}$, $g \! \in \! \{1,4,16\}$, $L \! \in \! \{1, 20, 400\}$.

In order to solve the ILP problems, we use COPT, a state-of-the-art commercial ILP solver~\cite{copt}. We use $5$ hours as the time limit for all the ILP problems. We run our experiments using an AMD EPYC 7763 processor (x86), with 1 TB memory and theoretical peak memory throughput of 204.8 GB/s.


The C++ implementation of our algorithms are available open-source 
on GitHub, as well as the new hypergraph benchmarks and the data from our experiments~\cite{folder}.

The datasets and the experimental results are discussed in more technical detail in Appendices~\ref{app:datasets} and \ref{app:exp}, respectively.

\subsection{Results for partitioning}

\begin{figure}[t]
    \centering
    \hspace{-4pt}
    \includegraphics[scale=0.54]{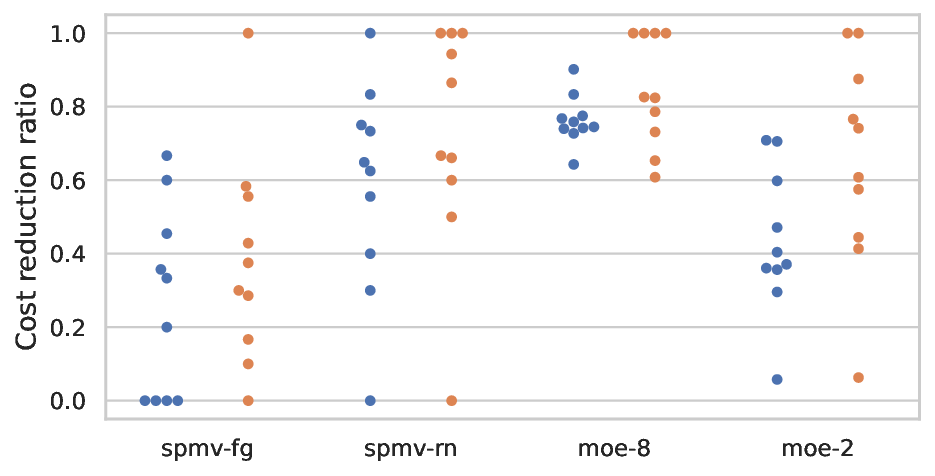}
    \vspace{-16pt}
    \caption{Distribution of cost reduction ratios on each dataset, for $P\!=\!2$, $\varepsilon\!=\!2.5\%$ (blue/left) and $P\!=\!4$, $\varepsilon\!=\!5\%$ (orange/right).}
    \label{fig:plot}
\end{figure}

For partitioning, we compare non-replication and replication via the ILPs in Section~\ref{sec:part}. We always run both ILP/D and ILP/R, and select the better solution found for the cost reduction ratio.

We focus on $P\!=\!2$ and $P\!=\!4$ for our main experiments. However, obtaining `comparable' results for $P\!=\!2$ and $P\!=\!4$ is a non-trivial question. Since each partition can have up to $\frac{n}{P} + \frac{\varepsilon}{P} \! \cdot \! n$ nodes, intuitively, the slack in a partition is proportional to $\frac{n}{P}$ instead of $n$. For $P\!=\!2$, this allows twice as many extra nodes in each partition as in $P\!=\!4$. As such, it is a more reasonable approach instead to keep the value $\hat{\varepsilon}:=\frac{\varepsilon}{P}$ fixed, i.e.\ to compare $\varepsilon\!=\!0.05$ for $P\!=\!4$ with $\varepsilon\!=\!0.025$ for $P\!=\!2$. We use these parameters in our experiments.

The cost reduction ratios for the instances are visualized in Figure~\ref{fig:plot}. The figure shows that replication can drastically decrease the partitioning cost. In particular, the optimal cost with replication actually goes down to $0$ for $5$ instances with $P\!=\!2$, and for $2$ instances with $P\!=\!4$. For the remaining instances, the mean cost reduction is $59.63\%$, $41.31\%$, $23.93\%$, $63.14\%$ on the datasets \texttt{spmv-fg}, \texttt{spmv-rn}, \texttt{moe-8} and \texttt{moe-2}, respectively, for $P\!=\!2$. The same cost reduction numbers are $65.22\%$, $21.93\%$, $16.99\%$ and $42.49\%$ for $P\!=\!4$. This confirms that replication can indeed have a remarkable impact on the optimal communication cost in these problems.

For $P\!=\!2$, the ILP solver finds the optimum in all cases. For $P\!=\!4$, the ILPs are much harder: in many cases, the solver cannot find the optimal solution with proof within the $5$-hour time limit. Specifically, for partitioning without replication, the solver could almost always prove optimality; with replication, this happened in significantly fewer cases. Nonetheless, if the non-replicating solution is optimal, then any better solution found with replication already establishes a lower bound on the gap between the two optimum values. Furthermore, our experience with the ILP solver suggests that the returned solutions are often (close to) optimal after a short time; the time-consuming part for the solver is to actually prove optimality. As such, we believe that our solutions are likely good indicators of the actual optimum cost in most cases.

In Appendix~\ref{app:exp}, we also compare the ILP/D and ILP/R formulations in detail based on how often they outperform each other and find optimal solutions. The results indicate that neither of them is clearly superior: ILP/R is stronger in general, but ILP/D is e.g.\ consistently better on \texttt{moe-8}.

We also note that in the special cases when the optimum cost with replication is actually $0$, algorithms for finding balanced vertex separators could also be a viable alternative to our ILPs.

Finally, we further discuss how the cost reduction depends on $P$ and $\varepsilon$. For this, we use a smaller ablation set of $12$ instances ($3$ from each dataset), all with $n \geq 160$. This reduces the running time of the experiments, and also ensures that the number of nodes allowed in a partition is different for any pair of $\varepsilon$ values we use, so we are indeed solving distinct problems.

For $P\!=\!2$, we compare the cost reduction for different $\varepsilon$ values in Table~\ref{tab:ablation_p2}. These instances are again always solved to optimality. The table shows that as $\varepsilon$ grows, the reductions become significantly larger: while the non-replicating optimum also decreases for larger $\varepsilon$, the impact of replication grows much more rapidly.

\begin{table}[t]
\caption{Mean cost reduction in partitioning for $P\!=\!2$ and different choices of $\varepsilon$ on the ablation set. The brackets show the number of times where cost was reduced to $0$; the mean cost reduction is taken on the remaining instances. \vspace{-8pt}}
  \begin{minipage}[b]{0.48\textwidth}
  \centering
    \renewcommand{\arraystretch}{1.45}
    \begin{tabular}{c || c | c | c | c | }
      & $\!\!$\texttt{spmv-fg}$\!\!$ & $\!\!$\texttt{spmv-rn}$\!\!$ & \texttt{moe-8} & \texttt{moe-2} \\ 
     \hline\hline
     $\!\!\varepsilon=1.25\%$ & 41.52\% (0) & 11.79\% (0) & 11.32\% (0) & 69.66\% (0) \\
     \hline
     $\!\!\varepsilon=2.5\%$ & 63.49\% (1) & 25.87\% (0) & 21.89\% (0) & 80.31\% (0) \\
     \hline
     $\!\!\varepsilon=5\%$ & 66.67\% (2) & 56.32\% (0) & 35.05\% (0) & 80.99\% (1) \\
     \hline
    \end{tabular}
  \end{minipage}
  \hfill
  \label{tab:ablation_p2}
\end{table}

We also run some experiments on the ablation set for larger $P$; we only summarize these here briefly, with the details deferred to Appendix~\ref{app:exp}. For $P\!=\!4$, we again study several $\varepsilon$ values; our results suggest that the mean reduction again grows with $\varepsilon$, although not so sharply. For $P\!=\!8$, we find that the problem is already too challenging for the ILP solver: reasonable solutions are only obtained for \texttt{spmv-fn} within the time limit. The reductions on these instances are almost identical to those of $P\!=\!4$.

\subsection{Results for scheduling}

For DAG scheduling, we analyze the impact of replication via our scheduling heuristics. Our results on the DAG datasets are summarized in Table~\ref{tab:sched_large_summary}. The two numbers in each cell describe the mean cost reduction achieved by the basic and the advanced heuristic, respectively, compared to the non-replicating initial schedule. The table shows that the algorithms, especially the advanced heuristic, can achieve significant cost reductions with replication, even from an already strong baseline schedule. For $P=8$, the improvement on the \texttt{hdb} and \texttt{psdd} datasets is $19.17\%$ and $23.13\%$, respectively, but the improvements on a single instance go up to $53.23\%$ and $33.62\%$. It is also visible that the improvement values grow larger for a higher number of processors; intuitively, parallelizing among more processors inherently requires more communication, and hence allowing replication has a higher impact altogether.

Table~\ref{tab:sched_g_and_l} shows the same numbers for $P\!=\!8$ with different choices of $g$ and $L$. These two BSP parameters directly describe the cost of communication, and hence it is not surprising that a higher $g$ or $L$ comes with higher gains from replication. For $g=16$, the cost reductions are $40.97\%$ and $56.49\%$ for \texttt{hdb} and \texttt{psdd}, respectively; for $L=400$, the same numbers are $58.17\%$ and $34.53\%$. Hence for larger communication costs, replication often allow us to reduce the total costs by more than a $2\times$ factor on average.

Interestingly, the improvements of the basic heuristic actually decrease for larger $L$ values. This is because synchronizations are a much larger portion of the total costs here, and the basic heuristic has no direct way to reduce the number of supersteps.

Note that \texttt{hdb} in fact consist of $3$ subgroups of different sizes: $1_{\,}000$--$2_{\,}000$, $5_{\,}000$--$10_{\,}000$, and $50_{\,}000$--$100_{\,}000$ nodes. Appendix~\ref{app:exp} shows the same results per subgroup, confirming that the improvements are rather consistent over different instance sizes.

Table~\ref{tab:sched_ablation} also analyzes each component of the advanced heuristic separately. This shows that both batch replication (BR) and superstep merging (SM) obtain a considerable further cost reduction from the basic heuristic, and hence they both contribute significantly to the final improvement. Superstep merging is in particular very useful for higher communication costs. On the other hand, the impact of the superstep replication (SR) is marginal at best.

\begin{table}[t]
\caption{Mean cost reduction with the basic / advanced heuristic for scheduling with replication, for $g=4$, $L=20$, organized by $P$ and the dataset. \vspace{-8pt}}
  \begin{minipage}[b]{0.48\textwidth}
  \centering
    \renewcommand{\arraystretch}{1.45}
    \begin{tabular}{c || c | c | c | c | }
     dataset & $P=2$ & $P=4$ & $P=8$ & $P=16$ \\ 
     \hline\hline
     \texttt{hdb} & \makecell{0.48\% /\\ 12.20\%} & \makecell{2.23\% /\\ 16.49\%} & \makecell{4.11\% /\\ 19.17\%} & \makecell{6.24\% /\\ 22.57\%} \\
     \hline
     \texttt{psdd} & \makecell{0.10\% /\\ 15.45\%} & \makecell{1.27\% /\\ 22.36\%} & \makecell{3.74\% /\\ 23.13\%} & \makecell{6.11\% /\\ 23.09\%} \\
     \hline
     \texttt{sptrsv} & \makecell{0.78\% /\\ 3.88\%} & \makecell{3.25\% /\\ 7.74\%} & \makecell{6.29\% /\\ 11.61\%} & \makecell{9.24\% /\\ 15.21\%} \\
     \hline
    \end{tabular}
  \end{minipage}
  \hfill
  \label{tab:sched_large_summary}
\end{table}

\begin{table}[t]
\caption{Mean cost reduction with the basic / advanced heuristic for scheduling with replication, for $P=8$, with different values of $g$ and $L$. \vspace{-8pt}}
  \begin{minipage}[b]{0.48\textwidth}
  \centering
    \renewcommand{\arraystretch}{1.45}
    \begin{tabular}{c || c || c | c || c | c | }
     dataset & \makecell{$g=4$ \\ $L=20$} & \makecell{$g=1$ \\ $L=20$} & \makecell{$g=16$ \\ $L=20$} & \makecell{$g=4$ \\ $L=1$} & \makecell{$g=4$ \\ $L=400$} \\ 
     \hline\hline
     \texttt{hdb} & \makecell{4.11\% /\\ 19.17\%} & \makecell{2.37\% /\\ 9.22\%} & \makecell{6.12\% /\\ 40.97\%} & \makecell{4.47\% /\\ 15.31\%} & \makecell{2.14\% /\\ 58.17\%} \\
     \hline
     \texttt{psdd} & \makecell{3.74\% /\\ 23.13\%} & \makecell{2.66\% /\\ 4.11\%} & \makecell{5.03\% /\\ 56.49\%} & \makecell{3.88\% /\\ 22.03\%} & \makecell{2.40\% /\\ 34.53\%}  \\
     \hline
     \texttt{sptrsv} & \makecell{6.29\% /\\ 11.61\%} & \makecell{2.58\% /\\ 5.56\%} & \makecell{12.7\% /\\ 26.45\%} & \makecell{7.10\% /\\ 10.97\%} & \makecell{3.10\% /\\ 39.18\%} \\
     \hline
    \end{tabular}
  \end{minipage}
  \hfill
  \label{tab:sched_g_and_l}
\end{table}

\begin{table}[t]
\caption{Ablation study: cost decrease obtained by activating only single components of the advanced heuristic, i.e.\ batch replication (BR), superstep merging (SM) and superstep replication (SR), compared to using only the basic heuristic (B). \vspace{-8pt}}
  \begin{minipage}[b]{0.48\textwidth}
  \centering
    \renewcommand{\arraystretch}{1.45}
    \begin{tabular}{c || c | c | c || c | c | c | }
     \multirow{2}{*}{$\!\!\!$heuristic$\!\!$} & \multicolumn{3}{|c|}{$P=8$,  $g=4$,  $L=20$} & \multicolumn{3}{|c|}{$P=8$,  $g=16$,  $L=20$} \\
     \hhline{~|------|}
      & \texttt{hdb} & \texttt{psdd} & $\!$\texttt{sptrsv}$\!$ & \texttt{hdb} & \texttt{psdd} & $\!$\texttt{sptrsv}$\!$ \\ 
     \hline\hline
     B & $\!$4.11\%$\!$ & $\!$3.74\%$\!$ & $\!$6.29\%$\!$ & $\!$6.12\%$\!$ & $\!$5.03\%$\!$ & $\!$12.71\%$\!$ \\
     \hline
     $\!\!$ B + BR & $\!$11.90\%$\!$ & $\!$18.38\%$\!$ & $\!$9.56\%$\!$ & $\!$21.26\%$\!$ & $\!$28.99\%$\!$ & $\!$20.33\%$\!$ \\
     \hline
     $\!\!$ B + SM & $\!$14.63\%$\!$ & $\!$21.06\%$\!$ & $\!$10.02\%$\!$ & $\!$38.69\%$\!$ & $\!$55.61\%$\!$ & $\!$23.64\%$\!$ \\
     \hline
     $\!\!$ B + SR & $\!$4.72\%$\!$ & $\!$3.77\%$\!$ & $\!$6.39\%$\!$ & $\!$7.03\%$\!$ & $\!$5.03\%$\!$ & $\!$12.87\%$\!$ \\
     \hline
    \end{tabular}
  \end{minipage}
  \hfill
  \label{tab:sched_ablation}
\end{table}

The results confirm that replication can drastically improve on heuristic schedules; however, it is not clear how these costs relate to optimality. Hence we also use ILPs to find the optimal schedules on some very small computational DAGs from~\cite{BSP_algos_opdas}. We run these experiments for $P \! \in \! \{2,4\}$, with $g\!=\!4$ and $L\!=\!5$; the mean cost reductions are $12.99\%$ and $21.08\%$ for $P=2$ and $P=4$, respectively. This is consistent with the results of the advanced heuristic on \texttt{hdb} and \texttt{psdd} in Table~\ref{tab:sched_large_summary}. More details are deferred to Appendix~\ref{app:exp}.

\section{Conclusion}

Our results demonstrate that replication can indeed significantly reduce the optimal cost in partitioning and scheduling problems. For partitioning, it can reduce the communication cost by $17\%$--$65\%$, or even allow to avoid all communication entirely. In scheduling, the cost can be reduced by $11.61\%$--$23.13\%$, or up to $58.17\%$ for some settings. This is especially impressive because allowing recomputation often comes at no other disadvantage or practical cost than requiring model flexibility: in partitioning, the maximal assignment to each partition remains unchanged, whereas in scheduling, the increased compute cost is already factored into these values.

The results also show that, unsurprisingly, the gain increases considerably as communication becomes a more dominant part of the problem, either directly via cost parameters or indirectly via the number of processors.

Besides the theoretical and experimental insights, our (advanced) heuristic for scheduling is a standalone contribution to utilize this approach for larger DAGs in applications. We believe that developing similar heuristics for replication in both partitioning and scheduling problems is one of the most interesting and impactful opportunities for future work.

\clearpage

\bibliographystyle{ACM-Reference-Format}
\bibliography{references}

\appendix

\section{Theoretical results}

We first discuss the technical details of the theoretical proofs.

\subsection{Example constructions for cost reduction} \label{app:reduction}

We begin with a more detailed discussion of the examples that show a large cost reduction with replication.

For partitioning, consider two cliques of size $\frac{1+\varepsilon}{2} \cdot n$ that intersect in $\varepsilon \cdot n$ nodes (we can take the floor of these value in case they are not integers). Without replication, any partitioning has a cost of at least $\Omega(n^2)$: indeed, the $\varepsilon \cdot n$ nodes in the intersection are adjacent to all other nodes, so from their outgoing edges, at least $\frac{1-\varepsilon}{2} \cdot n$ are cut in any solution. On the other hand, if we are allowed to replicate the $\varepsilon \cdot n$ common nodes, the cost of the optimal partitioning is $0$.

For scheduling, let $c$ be a constant parameter, and let us choose $P$ and $g$ as constants such that $g > P \! \cdot \! (P \! \cdot \! c + 1)$. Let $m=\frac{n}{c \cdot P + 1}$. Consider a DAG which consists of a set $U$ of $m$ nodes, and $P$ distinct sets $U_1$, ..., $U_P$ of $m \! \cdot \! c$ nodes each, such that for all $i \! \in \! [P]$, every node in $U$ has an edge to every node in $U_i$.

Without replication, the optimal schedule is to use a single processor and superstep, resulting in a total (compute-only) cost of $n = m \! \cdot \! (P \! \cdot \! c + 1)$. Indeed, if any two nodes from $U_1$, ..., $U_P$ are computed on different processors, then all nodes from $U$ must be sent to at least one other processor, resulting in at least $m$ communication steps, and hence a communication cost of at least $\frac{m \cdot g}{P}$. This is already larger than $m \! \cdot \! (P \! \cdot \! c + 1)$ since we chose $g > P \! \cdot \! (P \! \cdot \! c + 1)$. However, when $U_1$, ..., $U_P$ are assigned to the same processor $p$, the nodes in $U$ are also cheaper to assign to $p$: computing one of them on $p$ increases the compute cost by at most $1$, but decreases the communication costs (of values received by $p$) by $g > 1$.

On the other hand, with replication, we can compute $U$ on all processors, and assign a separate $U_i$ to each processor; this requires no communication, and its total (compute) cost in the single superstep is $(c+1) \! \cdot \! m$. As such, the ratio between the optimum with and without replication is $\frac{P \cdot c + 1}{c+1}$; with $c$ chosen large enough, this gets arbitrarily close to $P$.

\subsection{Proof of Theorem~\ref{th:part_inapprox}}

For the proof of Theorem~\ref{th:part_inapprox}, we use a reduction from the balanced vertex separator problem. In general, this problem requires us to partition $V$ into three disjoint sets $V_a$, $V_b$ and $V'$ such that there are no edges $(u,v) \in E$ such that $u \in V_a$ and $v \in V_b$, the size of the separator $|V'|$ is minimized, and the sizes $|V_a|$, $|V_b|$ satisfy some kind of a balance constraint. There are several variations of this balance constraint in the literature; in particular, many works on complexity and algorithms focus on a version where we must have $\max(|V_a|, |V_b|) \leq \alpha \cdot n$ after the separation for some constant $\alpha \geq \frac{1}{2}$, often chosen as $\alpha = \frac{2}{3}$. 

However, in our proof, we use the variant of the problem considered in e.g.~\cite{muller1991alpha}, which requires that $\min(|V_a|, |V_b|) \geq \frac{1}{2} \cdot (n-|V'|)$, or in other words, that $|V_a|=|V_b|$. It is shown in~\cite{muller1991alpha} that this problem is NP-hard, via a simple reduction from the balanced complete bipartite subgraph  (BCBS) problem. In particular, the problem of finding a balanced complete bipartite subgraph of size $k+k$ in a bipartite graph of size $n$ is equivalent to finding a separator $|V'|$ of size $n-2k$ in the complement graph.

When the NP-hardness of BCBS is discussed, this often happens without any assumptions on $k$~\cite{garey2002computers}. As such, we must begin by showing that for any $0 < \delta < \frac{1}{2}$, the BCBS problem is NP-hard for the specific choice of $k=\lceil \delta \cdot n \rceil$. Indeed, for any $0 < \delta < \frac{1}{2}$, assume we have a polynomial-time $\delta$-BCBS algorithm for finding a BCBS of size $\lceil \delta \cdot n \rceil$; we show how to use this to solve original problem of finding a BCBS for some specific size $k \geq 2$, which is known to be NP-hard. If we happen to have $\lceil \delta \cdot n \rceil < k$, then we can repeatedly add a pair of new nodes $v_1, v_2$ to the graph that are only connected to each other. These cannot be contained in a BCBS; furthermore, each such pair increases the value of $\delta \cdot n$ by less than $1$. Hence at some point we get $k=\lceil \delta \cdot n \rceil$, and then our $\delta$-BCBS algorithm will find a BCBS of size exactly $k$. On the other hand, if $\lceil \delta \cdot n \rceil > k$, we can repeatedly add a pair of new nodes $v_1, v_2$ that are connected to each other, and also $v_i$ is connected to all nodes in class $i$ for $i \in \{ 1, 2\}$. These can be included in any BCBS, and adding a new pair increases the size of our BCBS by $1$, while it increases $\delta \cdot n$ by strictly less than $1$. As such, after adding sufficiently many pairs $\ell$, we get that $k+\ell=\lceil \delta \! \cdot \! n \rceil$; a BCBS of size $(k+\ell)$ exists in the new graph exactly if a BCBS of size $k$ exists in the original one. Note that the number of pairs that need to be added in both cases is polynomial in $n$ for any given $\delta$. Hence for any $0 < \delta < \frac{1}{2}$, the problem of finding a BCBS of size at least $\delta \cdot n$ is NP-hard.

From here, the NP-hardness of partitioning with replication follows exactly as in the reduction of~\cite{muller1991alpha}. Specifically, we can assume for simplicity that in any partitioning with replication that both classes $V_1$ and $V_2$ contain exactly $\lfloor \frac{(1+\varepsilon)}{2}\cdot n \rfloor$ nodes, otherwise, we can insert arbitrary further nodes into the class without increasing the cost. In this case, balanced vertex separator  is essentially identical to partitioning with replication. In particular,  $V_1$, $V_2$ is a partitioning with replication that has cost $0$ if and only if $V_1 \setminus V_2$, $V_2 \setminus V_1$ and $V_1 \cap V_2$ form a balanced vertex separator as defined above: both $V_1 \setminus V_2$ and $V_2 \setminus V_1$ need to have at least $\lceil \frac{(1-\varepsilon)}{2}\cdot n \rceil$ nodes, and there can be no edges from $V_1 \setminus V_2$ to $V_2 \setminus V_1$ in order for the partitioning cost to be $0$.

As such, for some $\varepsilon \in (0, 1)$, assume we can decide in polynomial-time if a graph partitioning with replication and with cost $0$ exists. Let $\delta = \frac{(1-\varepsilon)}{2}$, and consider an instance of the $\delta$-BCBS problem. As in~\cite{muller1991alpha}, we can take the complement of the input graph; this complement has a partitioning with replication of cost $0$ exactly if the original graph has a $\delta$-BCBS. This completes the reduction.

Since it is already NP-hard to decide whether the optimum of the partitioning problem is $0$ or $1$, the inapproximability result follows.

\subsection{Proof of Theorem~\ref{th:additive}}

We know from the previous proof that for any $\varepsilon>0$, it is NP-hard to decide if there is a partitioning with replication of cost $0$. For Theorem~\ref{th:additive}, we provide a reduction from this previous problem.

Given a graph $G$ on $N$ nodes, we construct another graph $G'$ on $n$ nodes such that if the optimum (for partitioning with replication) in $G$ is $0$, then it is also $0$ in $G'$, but if the optimum is at least $1$ in $G$, then it is at least $n^{2-\delta}$ in $G'$. Through our previous NP-hardness result, this implies that it is NP-hard to approximate the optimum to an $n^{2-\delta}$ additive term.

Consider a small constant $\delta > 0$, and let us select $n=N^{1/\delta}$; this way, $n$ is still polynomial in $N$ for any $\delta > 0$. Let us replace every node $v$ in $G$ by a clique of size $n^{1-\delta}$; this way, the total number of nodes is $G'$ is indeed $N \cdot n^{1-\delta} = N^{(1-\delta)/\delta +1}=N^{1/\delta}$. For each edge $(u,v)$ in $G$, let us draw an edge from every node of the clique of $u$ to every node in the clique of $v$ in $G'$. Note that in this construction, any two nodes in a clique always have the exact same set of neighbors except for each other.

For ease of presentation, let us refer to the three sets $V_1 \cap V_2$, $V_1 \setminus V_2$ and $V_2 \setminus V_1$ as divisions. Let us say that a clique (corresponding to an original node of $G$) is unsplit if all its nodes are in the same division, and it is split otherwise.

Considering the original problem in $G$, let us use $B = \lceil \frac{1-\varepsilon}{2} \cdot N \rceil$ to denote the minimal number of nodes we need to place in both $V_1 \setminus V_2$ and $V_2 \setminus V_1$. Through most of the proof, we assume for simplicity that also in $G'$, the number of nodes we need to place in both $V_1 \setminus V_2$ and $V_2 \setminus V_1$ is $B \cdot n^{1-\delta}$; in other words, both $V_1 \setminus V_2$ and $V_2 \setminus V_1$ can fit exactly $B$ of our cliques in $G'$. Note that this is not necessarily the case due to the ceiling function, since we might have $\lceil \frac{1-\varepsilon}{2} \cdot N \cdot n^{1-\delta} \rceil \neq \lceil \frac{1-\varepsilon}{2} \cdot N \rceil \cdot n^{1-\delta}$, but we will address this technicality separately in the end.

It is not hard to see that the optimum is $0$ in $G$ exactly if it is $0$ in $G'$. If the optimum is $0$ in $G$, then we easily obtain a solution of cost $0$ in $G'$ by placing the corresponding cliques into the same divisions. On the other hand, assume we have a solution of cost $0$ in $G'$. Then due to the sizes of the divisions, we can select $B$ cliques for $V_1 \setminus V_2$ and $B$ cliques for $V_2 \setminus V_1$ such that the cliques have at least one node in the corresponding division. The two selections will also be disjoint: our partitioning in $G'$ has cost $0$, so no clique can intersect both divisions. Furthermore, there are also no edges from a clique selected for $V_1 \setminus V_2$ to a clique selected for $V_2 \setminus V_1$. Thus in $G$, this gives us $B+B$ nodes corresponding to the selected cliques with no edges from the nodes selected for $V_1 \setminus V_2$ to the nodes selected for $V_2 \setminus V_1$. Placing these nodes into $V_1 \setminus V_2$ and $V_2 \setminus V_1$ creates a solution of cost $0$ in $G$.

It remains to show that the optimum in $G'$ is either $0$ or at least $n^{2-\delta}$. We transform an arbitrary solution in $G'$ to a new solution where a smaller number of cliques is split, without increasing the cost. Assume first that there are two cliques $C_1$ and $C_2$ which both have nodes in both $V_1 \setminus V_2$ and $V_1 \cap V_2$. Assume without loss of generality that an (arbitrary) node in $C_1$ has at least as many edges towards $V_2 \setminus V_1$ as a node in $C_2$. We can then simultaneously move nodes of $C_1$ from $V_1 \setminus V_2$ to $V_1 \cap V_2$ and nodes of $C_2$ from $V_1 \cap V_2$ to $V_1 \setminus V_2$ without increasing the cost. Specifically, if $\ell = \min (|C_1 \cap (V_1 \setminus V_2)|, |C_2 \cap V_1 \cap V_2|)$, then we can move $\ell$ nodes of $C_1 \cap (V_1 \setminus V_2)$ to $V_1 \cap V_2$ and $\ell$ nodes of $C_2 \cap V_1 \cap V_2$ to $V_1 \setminus V_2$. This decreases the number of cliques split between $V_1 \setminus V_2$ and $V_1 \cap V_2$ by at least $1$.

We can continue this procedure until there is at most $1$ clique that intersects both $V_1 \setminus V_2$ and $V_1 \cap V_2$. We can do the same in a similar fashion to cliques split between $V_2 \setminus V_1$ and $V_1 \cap V_2$.

Similarly, assume there are two cliques $C_1$ and $C_2$ which both have nodes in both $V_1 \setminus V_2$ and $V_2 \setminus V_1$. Let $a_{1,1}$ and $a_{1,2}$ denote the number of edges from a node in $C_1$ to $V_1 \setminus (V_2 \cup C_1 \cup C_2)$ and $V_2 \setminus (V_1 \cup C_1 \cup C_2)$, respectively, and $a_{2,1}$ and $a_{2,2}$ denote the number of edges from a node in $C_2$ to $V_1 \setminus (V_2 \cup C_1 \cup C_2)$ and $V_2 \setminus (V_1 \cup C_1 \cup C_2)$, respectively. If $a_{1,1} - a_{1,2} -a_{2,1} + a_{2,2} \leq 0$, then we can again move nodes of $C_1$ from $V_1 \setminus V_2$ to $V_2 \setminus V_1$ and nodes of $C_2$ from $V_2 \setminus V_1$ to $V_1 \setminus V_2$, whereas if $a_{1,1} - a_{1,2} -a_{2,1} + a_{2,2} > 0$, we can move nodes of $C_1$ from $V_2 \setminus V_1$ to $V_1 \setminus V_2$ and nodes of $C_2$ from $V_1 \setminus V_2$ to $V_2 \setminus V_1$, as long as one of $C_1$ and $C_2$ has no more nodes in either $V_1 \setminus V_2$ or $V_2 \setminus V_1$. As for the edges between $C_1$ and $C_2$ in this process, note that there are either no edges from $C_1$ to $C_2$, or all nodes of $C_1$ are adjacent to all nodes of $C_2$; in the latter case, the number of cut edges that go from $C_1$ to $C_2$ remains unchanged in each move.

At this point, there is at most one clique that is split between any two specific divisions. There can be no clique $C$ that intersects all three divisions, because then every other clique should be unsplit, which is not possible if the sizes of all three divisions are integer multiples of the clique size. Only two possible cases remain.

One possibility is that we have no unsplit cliques at all. In this case, if there is a clique in $V_1 \setminus V_2$ and a clique in $V_2 \setminus V_1$ such that the corresponding nodes are adjacent in $G$, then this already results in $n^{1-\delta} \cdot n^{1-\delta}$ cut edges due to the complete bipartite connection. If there are no such cliques, then we have no cut edges at all, and the solution has cost $0$.

As the other possibility, there still might be three split cliques $C_1$, $C_2$, $C_3$ such that $C_1$ intersects $V_1 \setminus V_2$ and $V_2 \setminus V_1$, $C_2$ intersects $V_2 \setminus V_1$ and $V_1 \cap V_2$, and $C_3$ intersects $V_1 \cap V_2$ and $V_1 \setminus V_2$. In this case, for some $c_1$, $c_2$ with $c_1+c_2=n^{1-\delta}$, we have $c_1=|C_1 \cap (V_1 \setminus V_2)|=|C_2 \cap (V_2 \setminus V_1)|=|C_3 \cap (V_1 \cap V_2)|$ and $c_2=|C_1 \cap (V_2 \setminus V_1)|=|C_2 \cap (V_1 \cap V_2)|=|C_3 \cap (V_1 \setminus V_2)|$, once again because the set sizes are integer multiples of $n^{1-\delta}$. Assume w.l.o.g.\ that $c_1 \leq c_2$. If a node in $C_1 \cap V_1$ has at least as many edges towards $V_2 \setminus V_1$ as a node of $C_3$, then we can swap the nodes $C_1 \cap V_1$ and $C_3 \cap V_2$: both sets have cardinality $c_1$, the number of cut edges is not increased, and we end up with $C_3 \subseteq (V_1 \setminus V_2)$, and only two split cliques between $V_1 \cap V_2$ and $V_2 \setminus V_1$, which can then be removed with the methods above, again resulting in no unsplit cliques at all.

On the other hand, assume that a node in $C_3$ has more edges towards $V_2 \setminus V_1$ than a node in $C_1 \cap V_1$. With $c_1 \leq c_2$, both the size of $C_1 \cap V_2$ and the size of $C_3 \cap (V_1 \setminus V_2)$ is at least $\frac{1}{2} \cdot n^{1-\delta}$. Hence nodes in $C_1 \cap V_1$ have at least $\frac{1}{2} \cdot n^{1-\delta}$ edges towards $V_2 \setminus V_1$ (they are adjacent to $C_1 \cap V_2$), so nodes in $C_3$ also have at least $\frac{1}{2} \cdot n^{1-\delta}$ neighbors in $V_2 \setminus V_1$. With $C_3 \cap (V_1 \setminus V_2) \geq \frac{1}{2} \cdot n^{1-\delta}$, this is at least $\frac{1}{4} \cdot n^{2-2 \cdot \delta} = \Omega(n^{2-2\delta})$ cut edges.

Note that above we showed a factor of $n^{2-2\delta}$ or $\Omega(n^{2-2\delta})$, whereas the theorem has $n^{2-\delta}$ for any $\delta > 0$. However, both the factor $2$ in the exponent and the constant in the $\Omega$ can be removed with a choice of any smaller $\delta' < \frac{1}{2} \cdot \delta$ in the same proof, since we asymptotically have $\Omega(n^{2-2\delta'}) > n^{2-\delta}$.

Finally, recall that we assumed $\lceil \frac{1-\varepsilon}{2} \cdot N \cdot n^{1-\delta} \rceil = \lceil \frac{1-\varepsilon}{2} \cdot N \rceil \cdot n^{1-\delta}$ for convenience, whereas this is often not the case. To resolve this, we can select a different balance parameter $\varepsilon'$ such that $0 < \varepsilon' < \varepsilon $, and we in fact consider a reduction from the original problem with $\varepsilon'$; recall that Theorem~\ref{th:part_inapprox} holds for any $\varepsilon' \! \in \! (0,1)$. If we consider $n$ large enough, we can always ensure $\lceil \frac{1-\varepsilon}{2} \cdot N \cdot n^{1-\delta} \rceil < \lceil \frac{1-\varepsilon'}{2} \cdot N \rceil \cdot n^{1-\delta}$. We now add several extra nodes to the construction that are adjacent to every other node in $G'$; we do this until we have $\lceil \frac{1-\varepsilon}{2} \cdot n_0 \rceil = \lceil \frac{1-\varepsilon'}{2} \cdot N \rceil \cdot n^{1-\delta}$ for the new total number of nodes $n_0$. This is indeed always possible, since incrementing $n_0$ by $1$ increases $\frac{1-\varepsilon}{2} \! \cdot \! n_0$ by at most $\frac{1}{2}$. This ensures that exactly $\lceil \frac{1-\varepsilon'}{2} \! \cdot \! N \rceil$ cliques fit into both $V_1 \setminus V_2$ and $V_2 \setminus V_1$. Furthermore, we can assume that the extra nodes are always in $V_1 \cap V_2$; otherwise, swapping an extra node with a non-extra node in $V_1 \cap V_2$ never increases the number of cut edges. With all the extra nodes in $V_1 \cap V_2$, the problem is reduced to the simplified setting above: both $V_1 \setminus V_2$ and $V_2 \setminus V_1$ has enough space for exactly $\lceil \frac{1-\varepsilon'}{2} \cdot N \rceil$ cliques, and $V_1 \cap V_2$ has enough space exactly for the remaining cliques besides the extra nodes.

It only remains to discuss the effect of these extra nodes on $n_0$. For a naive bound, we will ensure with the choice of $\varepsilon'$ that we have $(1-\varepsilon') \leq 2 \cdot (1 - \varepsilon)$; after restructuring, this is equivalent to $\varepsilon' \geq 2 \! \cdot \! \varepsilon -1$. In case if $\varepsilon \leq \frac{1}{2}$, this holds automatically for any $\varepsilon' > 0$. Otherwise, if $\varepsilon > \frac{1}{2}$, we can specifically select $\varepsilon' = 2 \! \cdot \! \varepsilon \! - \! 1$. Furthermore, note that for $N$ large enough, we also have $\lceil \frac{1-\varepsilon'}{2} \cdot N \rceil \leq 2 \cdot \frac{1-\varepsilon'}{2} \cdot N$. Altogether, this implies that $\lceil \frac{1-\varepsilon'}{2} \cdot N \rceil \cdot n^{1-\delta} \leq 4 \cdot \frac{1-\varepsilon}{2} \cdot N \cdot n^{1-\delta}$; as such, the new total number of nodes is at most $n_0 \leq 4 \cdot N \cdot n^{1-\delta} = 4 \cdot n$ after adding the extra nodes. This means that in our proof, the number of cut edges becomes $n^{2-\delta} \geq \frac{1}{16} \cdot n_0\,\!^{2-\delta} = \Omega(n_0\,\!^{2-\delta})$. To remove the constant factor and obtain $n_0\,\!^{2-\delta}$, we can again apply the same proof with a choice of any $\delta'<\delta$.

\subsection{Proof of Lemma~\ref{lem:nphard}} \label{app:nphard}

Lemma~\ref{lem:nphard} claims that the properties for BSP scheduling on specific subclasses of DAGs in~\cite{BSP_DAG_opdas} also carry over to BSP scheduling with replication.

In DAGs with out-degree at most $1$, there is nothing to gain from replicating a node $v$. If $v$ has no out-neighbor, we can simply remove one of the replicas of $v$ without any effect. If $v$ has an out-neighbor $u$, we can assume that $u$ is not replicated (via an induction in a reverse topological order), but computed on a single processor $p$. If $v$ is computed on $p$, then its other replica can be discarded. If $v$ is computed on two distinct devices different from $p$, then $v$ can be discarded from the processor which sends the value of $v$ to $p$ in a later superstep (or does not send it at all); if both send it in the same superstep, we can discard an arbitrary one. As such, the problems in this case are essentially equivalent to scheduling without replication.

For the hardness results on $2$-layer DAGs, we do need to make some modifications to the original proof construction in~\cite{BSP_DAG_opdas}. Note that once again, the sink nodes are not replicated in any reasonable solution, but the source nodes might be. For instance, for the $2$-block gadgets that are a main ingredient of the construction, our new setting now only ensures that if the second-level blocks are split, then each node in the first level must be either sent \emph{or replicated}; as such, we need to increase the size of the first-level blocks appropriately to ensure that replicating each node from the block also comes with a prohibitively large cost. We have the same case for the first-level nodes $v_{e'}$ and $\hat{v}_i$ that might be needed on both processors: these can now either be sent or replicated. We pick any $g > 1$ to ensure that replication is the cheaper method, and we also modify the cost limit from $C_0 = \frac{n}{2} +
(|E'|- {k \choose 2}) \! \cdot \! g$ to $C_0 = \frac{n}{2} + (|E'|- {k \choose 2})$, so that besides the minimal work cost, it allows at most $(|E'|- {k \choose 2})$ replication steps on both processors. The arguments of the original proof then all carry over to the new setting, with communication steps replaced by replications. For more details, we refer the reader to the original proof in~\cite{BSP_DAG_opdas}.

\subsection{Proof of Theorem~\ref{thm:surplus}} \label{app:spes}

The proof construction for Theorem~\ref{thm:surplus} is also similar to the 2-layer DAG construction of~\cite{BSP_DAG_opdas}, but requires more significant changes. Here we focus on the modifications to the construction; for a detailed discussion of the original construction that motivates the idea, we again refer the reader to~\cite{BSP_DAG_opdas}.

We present the proof for BSP scheduling with replication; we discuss in the end how to also adapt it to BSP scheduling without replication.

The main building blocks of the construction are $2$-level block gadgets $(U^{(1)}, U^{(2)})$, which consist of a set $U^{(1)}$ of source nodes, a set $U^{(2)}$ of sink nodes, and edges $(u_1, u_2)$ for all $u_1 \! \in \! U^{(1)}$, $u_2 \! \in \! U^{(2)}$. For all blocks, we have $|U^{(1)}| \geq M_1$, $|U^{(2)}| \geq M_2$ for some large values $M_1$, $M_2$. Note that there is no reason to replicate any of the nodes in $U^{(2)}$, since they are sink nodes. Furthermore, in any reasonable solution, all nodes of $U^{(2)}$ are assigned to the same processor; otherwise, each node in $U^{(1)}$ must be either communicated or replicated, which results in a cost of at least $\frac{M_1}{2} \cdot \min(1, g)$, and the construction ensures that we can always trivially find a solution of lower cost. As such, we may assume that each $U^{(2)}$ is assigned entirely to a single processor.

In contrast to the proof in~\cite{BSP_DAG_opdas} which reduces from the clique problem, our modified construction provides a reduction from the smallest $k$-edge subgraph problem: given an undirected graph $G_0(V_0, E_0)$ on $N=|V_0|$ nodes, our goal is to find a minimal set of nodes such that their induced subgraph contains at least $k$ edges. The problem is known to be NP-hard to approximate to an $N^{1/(\log\log n)^{\delta'}}$ factor for some constant $\delta'>0$, assuming the exponential time hypothesis~\cite{ETHhardness}.

For every edge $e_0 \! \in \! E_0$ of the input graph $G_0$, we create a block $(U_{e_0}^{(1)}, U_{e_0}^{(2)})$ with $|U_{e_0}^{(1)}| = M_1$, $|U_{e_0}^{(2)}| = M_2$. For every node $v_0 \! \in \! V_0$, we create a source node $v$ in our construction, and draw an edge from $v$ to an arbitrarily selected node of $|U_{e_0}^{(2)}|$ for each edges $e_0$ that is incident to $v_0$ in $G_0$. We also add two blocks $(U_1^{(1)}, U_1^{(2)})$ and $(U_2^{(1)}, U_2^{(2)})$ to the construction, with $|U_1^{(1)}|=|U_2^{(1)}|=M_1$, but $|U_1^{(2)}|=(k+1) \! \cdot \! M_2$ and $|U_2^{(2)}|=(|E_0|- k + 1) \! \cdot \! M_2$. For each representative $v$ of a node $v_0 \in V_0$, we also draw an edge from $v$ to an arbitrary node in $U_1^{(2)}$. Finally, we add $|E_0| \! \cdot \! M_1 + |V_0|$ isolated nodes to finish the construction. 

The intended BSP schedule for the DAG is as follows. We can assign $(U_i^{(1)}, U_i^{(2)})$ to processor $i$ for $i \! \in \! \{ 1,2 \}$. We can assign $|E_0|-k$ of the blocks $(U_{e_0}^{(1)}, U_{e_0}^{(2)})$ to processor $1$ and the remaining $k$ to processor $2$; this already ensures that the second blocks $U^{(2)}$ of the block gadget altogether contribute the same amount of nodes to both processors. We then assign all sources $v$ that correspond to a node $v_0 \! \in \! V_0$ to processor $1$, and if $v_0$ has an incident edge gadget assigned to processor $2$, then we also replicate $v$ on processor $2$. Assume there are $d$ nodes replicated on processor $2$ this way. Finally, we split the isolated nodes between the processors to ensure that the node assignment is as balanced as possible: we assign $k \cdot M_1 + \lceil \frac{d}{2} \rceil$ isolated nodes to processor $1$, and $(|E_0| - k) \cdot M_1 + |V_0| - \lceil \frac{d}{2} \rceil$ isolated nodes to processor $2$. This schedule consists of only a single superstep, requires no communication at all, and its surplus cost is $\lceil \frac{d}{2} \rceil$.

The appropriate choice of parameters ensures that any reasonable schedule is of the form above. Indeed, $M_1 > |V_0| +1 > d$ and $g>1$ ensures that splitting a sink group $U^{(2)}$, as discussed above, results in a higher surplus cost than the above schedule with any choice of $k$ edges. Let $M_2'=2 \cdot (|E_0| + 2) \cdot M_1 + 2 \cdot |V_0|$ denote the number of nodes in the construction besides the sink groups $U^{(2)}$. We then select $M_2 > M_2' + 2 \! \cdot \! |V_0| + 2$. This ensures that the sink groups $U^{(2)}$ altogether must add up to exactly $(|E_0|+1) \cdot M_2$ nodes on both processors; otherwise, the compute cost is at least $2 \cdot (|V_0|+1)$ higher on one of the processors, so we again have a surplus cost of at least $|V_0| + 1 > d$. Similarly, we cannot assign $U_1^{(2)}$ and $U_2^{(2)}$ to the same processor due to their large size: even if all other nodes are assigned to the other processor, this would have a surplus cost of $\frac{|M_2|}{2} > d$ at least. As such, the only way to balance the sink groups $U^{(2)}$ without splitting a group is to assign $U_1^{(2)}$ and $(|E_0|-k)$ edge gadgets to one processor (assume w.l.o.g. that this is processor $1$), and assign $U_2^{(2)}$ and the remaining $k$ edge gadgets to processor $2$.

This already implies that the value of all source vertices $v$ (that correspond to a $v_0 \in V_0$) must eventually be present on processor $1$, and if $d$ of them have incident edges $e_0$ with the corresponding gadget assigned to processor $2$, then these $d$ values must also be eventually present on processor $2$. Whether this happens with replication or communication, this results in a surplus cost of at least $\lceil \frac{d}{2} \rceil$. On the other hand, a surplus cost of only $\lceil \frac{d}{2} \rceil$ is always obtainable in a single superstep as discussed above, by keeping $U^{(1)}$ and $U^{(2)}$ on the same processor for all blocks, assigning all nodes $v$ to processor $1$ and replicating the $d$ required ones on processor $2$, and then assigning the appropriate number of isolated nodes to each processor to maintain work balance. This means that we can assume that all our schedules are in this form; if an approximation algorithm to BSP scheduling (with or without replication) returns a schedule that does not satisfy these properties, we can simply replace it by a schedule fulfilling these properties, and this will have even lower (surplus) cost.

This implies that if the optimum cost of the input smallest $k$-edge subgraph problem is $\text{OPT}_{SkES}$, then the corresponding BSP scheduling problem will have optimal surplus cost $\text{OPT}_{surp} = \lceil \frac{\text{OPT}_{SkES}}{2} \rceil$, and hence in particular, we have $\text{OPT}_{SkES} \geq \text{OPT}_{surp}$. Similarly, if an approximation algorithm for BSP scheduling returns a schedule with surplus cost $\text{SOL}_{surp}$, then by selecting the edge gadgets assigned to processor $2$, we can convert this into a smallest $k$-edge subgraph solution on at most $\text{SOL}_{SkES} \leq 2 \cdot \text{SOL}_{surp}$ nodes. Also note that with our choice of parameters, we have $M_1=O(N)$ and $M_2=O(N^3)$; the number of nodes in the construction is altogether $n=O(N^5)$, so there is a small constant $c$ such that $n \leq c \! \cdot \! N^5$.

Recall that $\delta'$ denotes the constant from the inapproximability result for smallest $k$-edge subgraph~\cite{ETHhardness}, and let us select a constant $\delta$ such that $\delta > \delta'$. Assume we have a polynomial-time approximation algorithm for the surplus cost in BSP scheduling to a $\frac{\text{SOL}_{surp}}{\text{OPT}_{surp}} \leq n^{1/(\log\log n)^{\delta}}$ factor. Then for the corresponding solution to smallest $k$-edge subgraph, we have that
\[ \frac{\text{SOL}_{SkES}}{\text{OPT}_{SkES}} \leq \frac{2 \cdot \text{SOL}_{surp}}{\text{OPT}_{surp}} \leq 2 \cdot n^{1/(\log\log n)^{\delta}} \, . \]
With $n \leq c \! \cdot \! N^5$ and $N \leq n$, the right-hand side can be further upper bounded by $2 \! \cdot \! c \! \cdot \! N^{5/(\log\log N)^{\delta}}$. Since $\delta > \delta'$, there exists an $N$ large enough such that $2 \! \cdot \! c \! \cdot \! N^{5/(\log\log N)^{\delta}} < N^{1/(\log\log N)^{\delta'}}$. As such, we have obtained an $N^{1/(\log\log N)^{\delta'}}$-factor approximation algorithm for the smallest $k$-edge subgraph problem. This completes the reduction.

If we instead consider the BSP scheduling problem without replication, then we only need to replace replication steps by communication steps in the proof to hold in an identical way. Specifically, let us now select $g=1$ and $L=0$. The source nodes that needed to be replicated before now need to be communicated to the other processor. A reasonable solution distributes these in a balanced way among the two processors, so the single communication phase incurs a cost of $g \! \cdot \! \lceil \frac{d}{2} \rceil = \lceil \frac{d}{2} \rceil$, which is the same surplus cost as before. With this, the arguments above all carry over to the non-replicating case, too.

\section{Benchmarks for the experiments} \label{app:datasets}

We now discuss the experiment benchmarks in more detail.

\subsection{For partitioning}

\subsubsection*{MoE hypergraphs}

The hypergraph benchmark for MoE optimization is novel in this paper. We use open-source dataset of~\cite{moe_data}, which profiles the expert usage of modern MoE frameworks on several benchmarks. The dataset offers two Mixture-of-Experts models: Qwen (FP8 quantized Qwen3-235B) on 128 experts per layer, and DeepSeek (4 bit quantized Deepseek R1 using AWQ) on 256 experts per layer. From the datasets, we use the Massive Multitask Language Understanding (MMLU) benchmark, which measures the capabilities of the LLMs in 57 different subjects, ranging from natural sciences to social sciences and humanities. For both models, data is available for both the prefill and the decode phase; we focus on the data from the decode phase, since the optimization of this is often more relevant. The benchmark altogether provides the expert access pattern of more than 1.5 million tokens for DeepSeek and almost 1.8 million tokens for Qwen, thus allowing to build rather accurate hypergraph models of expert usage.

We form hypergraphs in the \texttt{moe-8} dataset as follows. We consider the $8$-tuples of experts that each of the 1.5+ millions of tokens invokes on a specific layer, and we select the $8$-tuples that appear most frequently. We select a pin limit parameter $\kappa_0$, and we find the smallest integer $f$ such that there are at least $\frac{\kappa_0}{8}$ distinct $8$-tuples appearing at least $f$ times. We create a hyperedge for each $8$-tuple that appears at least $f$ times; this ensures that the number of pins in the hypergraph is $\kappa_0$ or only slightly above. For the weight of each hyperedge, we consider the frequency of the given $8$-tuple and normalize it into the range $[1, 10]$, with values below $1$ rounded up to $1$. On the other hand, the weights of the nodes are uniformly $1$ here.

For the \texttt{moe-2} dataset, we follow the same procedure, but instead of the $8$-tuples, we consider all the subsets of size $2$ in each $8$-tuple for the given layer. The interpretation of the pin limit $\kappa_0$ is adjusted accordingly: we now find the smallest integer $f$ such that there are at least $\frac{\kappa_0}{2}$ distinct pairs appearing at least $f$ times. Focusing on subsets of size $2$ actually creates a simple graph, which is a special case of a hypergraph; as such, this also provides some insight into the relative hardness of graph partitioning to hypergraph partitioning for our ILP solver. Note that the \texttt{moe-2} dataset is also motivated by many practical approaches in this area, which also often choose to focus only on pairs of frequently co-occurring experts instead of larger subsets~\cite{moetuner}.

With the methodology above, our hypergraphs will often contain isolated nodes when a specific experts does not appear in any of the $f$ most frequent hyperedges. Since these isolated nodes might make the partitioning problem substantially easier, we remove them to obtain connected hypergraphs for our experiments. However, we point out that of course in a practical application of the method, these experts also need to be placed on some device. The accurate modeling of this setting would ideally require a more sophisticated hypergraph model where the connection of these nodes to the rest of the experts is also represented, but without resulting a significant increase of the number of hyperedges or pins. We leave the exploration of these more advanced hypergraph models to domain-specific future work.

Our goal is hence to select a relatively small $\kappa_0$ so that the number of pins is similar to the other dataset, but we still want this to cover a high portion of the original experts. Our analysis of the data shows that the first layers of the LLMs are much more suitable for this: the same choice of $\kappa_0$ covers much more experts here. As such, for both Qwen and DeepSeek, we select the first $5$ layers of the model to form a hypergraph in our test set. In Qwen, this simply means the layers indexed $0$-$4$. In DeepSeek, layers $0$-$2$ do not actually contain experts yet, so here the first $5$ expert layers actually corresponds to layers $3$--$7$; to avoid confusion, we also re-index these to the range $0$-$4$.

For our hypergraphs, we choose $\kappa_0=1000$. For the case of Qwen, this already ensures that in each of the first $5$ layers, even with the restriction to pairs in \texttt{moe-8}, we cover at least $80$ of the $128$ experts. Note compared to the smaller SpMV hypergraphs, this generally produces instances with a smaller number of nodes but a higher number of pins. For DeepSeek, we use the same choice of $\kappa_0=1000$, which produces hypergraphs with a similar number of pins, but more nodes and somewhat smaller average degree. We point out that we also experimented with $\kappa_0=2000$ with DeepSeek in order to obtain some even larger hypergraphs; however, the ILP solver struggled to find reasonable solutions for these within our $5$-hour time limit.

The properties of the instances are also summarized in Table~\ref{tab:moe_inst}. Note that the primary aspect for the sorting of the table rows is not the dataset that that the instance is contained in (\texttt{moe-8} or \texttt{moe-2}), but the LLM model where it originates from.

\begin{table}[t]
\caption{Properties of MoE hypergraphs: name, model ($8$-tuples or pairs), and number of nodes, hyperedges, pins.}
  \begin{minipage}[b]{0.48\textwidth}
  \centering
    \renewcommand{\arraystretch}{1.35}
    \begin{tabular}{c | c || c | c| c |}
     name & model & nr. nodes & \makecell{nr. hyper- \\ edges} & nr. pins \\ 
     \hline\hline
     Qwen\_l0 & \texttt{moe-8} & 109 & 125 & 1000 \\
     \hline
     Qwen\_l1 & \texttt{moe-8} & 94 & 125 & 1000 \\ 
     \hline
     Qwen\_l2 & \texttt{moe-8} & 88 & 125 & 1000 \\ 
     \hline
     Qwen\_l3 & \texttt{moe-8} & 81 & 126 & 1008 \\ 
     \hline
     Qwen\_l4 & \texttt{moe-8} & 86 & 126 & 1008 \\ 
     \hline
     Qwen\_l0 & \texttt{moe-2} & 97 & 500 & 1000 \\ 
     \hline
     Qwen\_l1 & \texttt{moe-2} & 92 & 500 & 1000 \\ 
     \hline
     Qwen\_l2 & \texttt{moe-2} & 84 & 500 & 1000 \\ 
     \hline
     Qwen\_l3 & \texttt{moe-2} & 80 & 501 & 1002 \\ 
     \hline
     Qwen\_l4 & \texttt{moe-2} & 80 & 500 & 1000 \\
     \hline
     \hline
     DeepSeek\_l0 & \texttt{moe-8} & 187 & 125 & 1000 \\
     \hline
     DeepSeek\_l1 & \texttt{moe-8} & 187 & 126 & 1008 \\
     \hline
     DeepSeek\_l2 & \texttt{moe-8} & 201 & 125 & 1000 \\
     \hline
     DeepSeek\_l3 & \texttt{moe-8} & 215 & 126 & 1008 \\
     \hline
     DeepSeek\_l4 & \texttt{moe-8} & 212 & 125 & 1000 \\
     \hline
     DeepSeek\_l0 & \texttt{moe-2} & 149 & 500 & 1000 \\
     \hline
     DeepSeek\_l1 & \texttt{moe-2} & 141 & 500 & 1000 \\
     \hline
     DeepSeek\_l2 & \texttt{moe-2} & 165 & 500 & 1000 \\
     \hline
     DeepSeek\_l3 & \texttt{moe-2} & 161 & 502 & 1004 \\
     \hline
     DeepSeek\_l4 & \texttt{moe-2} & 199 & 500 & 1000 \\
     \hline
    \end{tabular}
  \end{minipage}
  \hfill
  \label{tab:moe_inst}
\end{table}

\subsubsection*{SpMV hypergraphs}

For SpMV, we use the SuiteSparse matrix collection~\cite{davis2011university}, a popular benchmark for sparse matrix computations, available at \url{https://sparse.tamu.edu/}. For both the fine-grained and the row-net case, we select 10 matrices of varying size ranging from $90$ to $500$ nodes. Our goal was to select matrices that originate from a diverse set of application areas, while also filtering out matrices with heavily structured nonzero patterns where the best partitioning would be trivial to find. We also made sure to include some of the larger matrices from the experiments of~\cite{jenneskens22}.

In the fine-grained model, each node represents a non-zero, and as such their compute weight is uniformly $1$. In the row-net model, the weight of each vertex (corresponding to a column) is defined as the number of non-zero entries in the given column~\cite{catalyurek2002hypergraph}. The weight of each hyperedge is uniformly $1$ in both of the models.

The properties of the instances are summarized in Table~\ref{tab:spmv_inst}.

\begin{table}[t]
\caption{Properties of SpMV hypergraphs: name, model (Fine-Grained or Row-Net), and number of nodes, hyperedges, pins.}
  \begin{minipage}[b]{0.48\textwidth}
  \centering
    \renewcommand{\arraystretch}{1.35}
    \begin{tabular}{c | c || c | c| c |}
     name & model & nr. nodes & \makecell{nr. hyper- \\ edges} & nr. pins \\ 
     \hline\hline
     can\_24 & \texttt{spmv-fg} & 92 & 48 & 184 \\ 
     \hline
     bcspwr02 & \texttt{spmv-fg} & 108 & 98 & 216 \\
     \hline
     refine & \texttt{spmv-fg} & 153 & 91 & 306 \\ 
     \hline
     mesh1e1 & \texttt{spmv-fg} & 177 & 96 & 354 \\ 
     \hline
     d\_dyn & \texttt{spmv-fg} & 238 & 174 & 476 \\ 
     \hline
     bcspwr03 & \texttt{spmv-fg} & 297 & 236 & 594 \\ 
     \hline
     west0156 & \texttt{spmv-fg} & 371 & 312 & 742 \\ 
     \hline
     impcol\_c & \texttt{spmv-fg} & 411 & 274 & 822 \\ 
     \hline
     west0132 & \texttt{spmv-fg} & 414 & 264 & 828 \\
     \hline
     bfwa62 & \texttt{spmv-fg} & 450 & 124 & 900 \\ 
     \hline
     \hline
     bcspwr03 & \texttt{spmv-rn} & 118 & 118 & 297 \\
     \hline
     impcol\_c & \texttt{spmv-rn} & 137 & 137 & 411 \\
     \hline
     west0156 & \texttt{spmv-rn} & 156 & 156 & 371 \\
     \hline
     rajat14 & \texttt{spmv-rn} & 180 & 180 & 1503 \\
     \hline
     impcol\_e & \texttt{spmv-rn} & 225 & 225 & 1308 \\
     \hline
     lshp\_265 & \texttt{spmv-rn} & 265 & 265 & 1009 \\
     \hline
     bcspwr04 & \texttt{spmv-rn} & 274 & 274 & 973 \\
     \hline
     lp\_lotfi & \texttt{spmv-rn} & 366 & 153 & 1136 \\
     \hline
     west0479 & \texttt{spmv-rn} & 479 & 479 & 1910 \\
     \hline
     494\_bus & \texttt{spmv-rn} & 494 & 494 & 1080 \\
     \hline
    \end{tabular}
  \end{minipage}
  \hfill
  \label{tab:spmv_inst}
\end{table}

\subsection{For scheduling}

\subsubsection*{The \texttt{hdb} dataset}

The \texttt{hdb} dataset uses the computational DAG database introduced by Papp et al. for studying scheduling algorithms in the BSP model~\cite{BSP_algos_opdas}. The dataset consists of 5 smaller datasets, named tiny (40-80 nodes), small (250-600 nodes), medium (1000-2000 nodes), large (5000-10000 nodes) and huge (50000-100000 nodes). For the \texttt{hdb} dataset, we combine the medium, large and huge datasets, thus altogether obtaining $52$ computational DAGs with the number of nodes between $1000$--$100000$. The dataset is available at \url{https://github.com/Algebraic-Programming/Artifacts/tree/master/SPAA_2024_Efficient_Multi-Processor_Scheduling/database/used_for_experiments/test_set}.

Most of the computational DAGs in this dataset are derived from a fine-grained representation of four fundamental computations (SpMV, conjugate gradient, $k$-nearest neighbors, iterated matrix multiplication) in algebraic form, using a matrix with a randomized non-zero pattern. There are also a few DAGs that correspond to a coarse-grained representation of computations from similar areas. The DAGs in the dataset also contain node weights for the $4$ fine-grained groups; for the few coarse-grained graphs, these are uniformly $1$. 
We summarize the number of nodes and directed edges in each of these DAGs in Table~\ref{tab:hdb_inst}.

\begin{table}[t]
\caption{Properties of computational DAGs in the \texttt{hdb} dataset: name, as well as number of nodes and edges. The original names in the dataset are replaced by short identifiers for the sake of brevity.}
    \hspace{-0.03\textwidth}
    \begin{minipage}{0.2\textwidth}
        \centering
        \renewcommand{\arraystretch}{1.35}
        \vspace{-14pt}
        \begin{tabular}{c | c | c |}
            DAG & \makecell{nr.\\$\!\!$nodes$\!\!$} & \makecell{nr. \\$\!\!$edges$\!\!$} \\ 
         \hline\hline
          $\!$exp\_N35\_K4$\!$ & 1020 & $\!$2028$\!$ \\ \hline
          $\!$CG\_N12\_K6$\!$ & 1067 & $\!$2179$\!$ \\ \hline
          $\!\!$kNN\_N40\_K5$\!$ & 1073 & $\!$2007$\!$ \\ \hline
          $\!$exp\_N30\_K6$\!$ & 1085 & $\!$2268$\!$ \\ \hline
          $\!$CG\_N9\_K9$\!$ & 1099 & $\!$2242$\!$ \\ \hline
          $\!$spmv\_N60$\!$ & 1170 & $\!$1575$\!$ \\ \hline
          $\!$kNN\_N30\_K8$\!$ & 1219 & $\!$2601$\!$ \\ \hline
          $\!$exp\_N30\_K8$\!$ & 1395 & $\!$3024$\!$ \\ \hline
          $\!$CG\_N15\_K7$\!$ & 1512 & $\!$3116$\!$ \\ \hline
          $\!$CG\_N12\_K10$\!$ & 1527 & $\!$3089$\!$ \\ \hline
          $\!\!$kNN\_N50\_K4$\!$ & 1541 & $\!$2763$\!$ \\ \hline
          $\!\!\!$kNN\_N30\_K10$\!\!$ & 1585 & $\!$3519$\!$ \\ \hline
          $\!$spmv\_N65$\!$ & 1604 & $\!$2211$\!$ \\ \hline
          $\!$exp\_N40\_K5$\!$ & 1632 & $\!$3480$\!$ \\ \hline
          $\!$CG\_N21\_K5$\!$ & 1888 & $\!$3996$\!$ \\ \hline
          $\!$kNN\_N50\_K5$\!$ & 1896 & $\!$3801$\!$ \\ \hline
          $\!$spmv\_N70$\!$ & 1920 & $\!$2670$\!$ \\ \hline
          $\!$exp\_N44\_K5$\!$ & 1932 & $\!$4170$\!$ \\ \hline
          $\!$exp\_N30\_K10$\!$ & 1950 & $\!$4443$\!$ \\ \hline
          $\!\!\!$kNN\_N30\_K12$\!\!$ & 1951 & $\!$4437$\!$ \\ \hline
          $\!$CG\_N17\_K8$\!$ & 1976 & $\!$4134$\!$ \\ \hline
          $\!\!\!$kNN\_N70\_K7$\!\!$ & 5228 & $\!\!$11991$\!\!$ \\ \hline
          $\!$spmv\_N120$\!$ & 5326 & $\!$7629$\!$ \\ \hline
          $\!\!\!$kNN\_N45\_K15$\!\!$ & 5392 & $\!\!$13275$\!\!$ \\ \hline
          $\!$CG\_N35\_K8$\!$ & 5396 & $\!\!$12153$\!\!$ \\ \hline
          $\!$exp\_N50\_K12$\!$ & 5434 & $\!\!$13248$\!\!$ \\ \hline  
        \end{tabular}
    \end{minipage}
    \hspace{0.03\textwidth}
    \begin{minipage}{0.2\textwidth}
        \centering
        \renewcommand{\arraystretch}{1.35}
        \begin{tabular}{c | c | c |}
            DAG & \makecell{nr.\\$\!\!$nodes$\!\!$} & \makecell{nr. \\$\!\!$edges$\!\!$} \\ 
         \hline\hline
          $\!\!$exp\_N70\_K5$\!\!$ & $\!$5550$\!$ & $\!$12825$\!$ \\ \hline
          $\!\!$CG\_N25\_K17$\!\!$ & $\!$5658$\!$ & $\!$12072$\!$ \\ \hline
          $\!\!\!$kNN\_N55\_K14$\!\!$ & $\!$6941$\!$ & $\!$17244$\!$ \\ \hline
          $\!\!$exp\_N80\_K6$\!\!$ & $\!$7056$\!$ & $\!$16704$\!$ \\ \hline
          $\!\!$exp\_N55\_K14$\!\!$ & $\!$7230$\!$ & $\!$17934$\!$ \\ \hline
          $\!\!$CG\_N40\_K10$\!\!$ & $\!$7335$\!$ & $\!$16632$\!$ \\ \hline
          $\!\!$CG\_N24\_K22$\!\!$ & $\!$7383$\!$ & $\!$16145$\!$ \\ \hline
          $\!\!$spmv\_N130$\!\!$ & $\!$7482$\!$ & $\!$10833$\!$ \\ \hline
          $\!\!$kNN\_N75\_K9$\!\!$ & $\!$7550$\!$ & $\!$18231$\!$ \\ \hline
          $\!\!\!$kNN\_N60\_K15$\!\!$ & $\!$8998$\!$ & $\!$22722$\!$ \\ \hline
          $\!\!$spmv\_N150$\!\!$ & $\!$9088$\!$ & $\!$13182$\!$ \\ \hline
          $\!\!$kNN\_N90\_K8$\!\!$ & $\!$9108$\!$ & $\!$21717$\!$ \\ \hline
          $\!\!$exp\_N60\_K17$\!\!$ & $\!$9125$\!$ & $\!$22821$\!$ \\ \hline
          $\!\!$exp\_N90\_K7$\!\!$ & $\!$9632$\!$ & $\!$23394$\!$ \\ \hline
          $\!\!$CG\_N45\_K13$\!\!$ & $\!$9786$\!$ & $\!$22088$\!$ \\ \hline
          $\!\!$CG\_N30\_K24$\!\!$ & $\!$9941$\!$ & $\!$21754$\!$ \\ \hline
          $\!\!$cg\_gyro\_m$\!\!$ & $\!\!$47023$\!\!$ & $\!$70893$\!$ \\ \hline
          $\!\!\!$bicgstab\_gyro$\!\!\!$ & $\!\!$50195$\!\!$ & $\!$80373$\!$ \\ \hline
          $\!\!$exp\_N75\_K60$\!\!$ & $\!\!$54229$\!\!$ & $\!\!$146520$\!\!$ \\ \hline
          $\!\!$spmv\_N300$\!\!$ & $\!\!$54300$\!\!$ & $\!$80550$\!$ \\ \hline
          $\!\!$CG\_N60\_K55$\!\!$ & $\!\!$54324$\!\!$ & $\!\!$127658$\!\!$ \\ \hline
          $\!\!\!$kNN\_N100\_K50$\!\!$ & $\!\!$56521$\!\!$ & $\!\!$151731$\!\!$ \\ \hline
          hpcg & $\!\!$72685$\!\!$ & $\!\!$193610$\!\!$ \\ \hline
          $\!\!$exp\_N100\_K40$\!\!$ & $\!\!$83394$\!\!$ & $\!\!$232080$\!\!$ \\ \hline
          $\!\!\!$kNN\_N250\_K15$\!\!$ & $\!\!$94667$\!\!$ & $\!\!$253986$\!\!$ \\ \hline
          $\!\!$CG\_N120\_K35$\!\!$ & $\!\!$96685$\!\!$ & $\!\!$241842$\!\!$ \\ \hline 
        \end{tabular}
        \vspace{14pt}
    \end{minipage}
  \label{tab:hdb_inst}
\end{table}

To briefly analyze the relation of the heuristic costs to optimality, our scheduling ILP experiments for also consider the so-called tiny dataset from~\cite{BSP_algos_opdas}, which contains $16$ computational DAGs of $40$--$80$ nodes, with very similar structure and origin to \texttt{hdb}.

\subsubsection*{The \texttt{psdd} dataset}

The \texttt{psdd} dataset originates from the Circuit Model Zoo of the UCLA Star AI Lab; Shah et al. use this dataset for experiments in their recent work on DAG-specific compilation~\cite{shah2022dpu}. The dataset is available at \url{https://github.com/UCLA-StarAI/Circuit-Model-Zoo/tree/master/psdds}. Originally, the dataset also contains some very small example DAGs; we remove these, restricting ourselves to DAGs above $1000$ nodes. This results in a dataset of $29$ instances; $26$ of these are between $2686$ and $40034$ nodes, and there are three larger instances with $63488$, $89389$ and $174135$ nodes, respectively. The properties of these computational DAGs are summarized in Table~\ref{tab:psdd_inst}.

\begin{table}[t]
\caption{Properties of computational DAGs in the \texttt{psdd} dataset: name, number of nodes and edges in the DAG.}
    \begin{minipage}{0.2\textwidth}
        \centering
        \renewcommand{\arraystretch}{1.35}
        \begin{tabular}{c | c | c |}
            DAG & \makecell{nr.\\$\!\!$nodes$\!\!$} & \makecell{nr. \\$\!\!$edges$\!\!$} \\ 
         \hline\hline
          kdd-6k & 2686 & 6527 \\
         \hline
          tretail & 2822 & 6382 \\
         \hline
          $\!\!$elevators$\!$ & 3331 & 6892 \\
         \hline
          \makecell{msnbc-\\yitao-d} & 4240 & 6908 \\
         \hline
          nltcs & 4676 & 10592 \\
         \hline
          mnist & 4904 & 5916 \\
         \hline
          \makecell{mnist-\\antonio} & 4904 & 5916 \\
         \hline
          \makecell{msnbc-\\yitao-c} & 4989 & 7164 \\
         \hline
          \makecell{msnbc-\\yitao-e} & 5216 & 8074 \\
         \hline
          \makecell{msnbc-\\yitao-b} & 6590 & 8870 \\
         \hline
          \makecell{msnbc-\\yitao-a} & 6682 & 8926 \\
         \hline
          $\!\!$insurance$\!$ & $\!$14342$\!$ & $\!$30060$\!$ \\
         \hline
          bnetflix & $\!$15340$\!$ & $\!$40736$\!$ \\
         \hline
         msweb & $\!$15777$\!$ & $\!$38289$\!$ \\
         \hline
        \end{tabular}
    \end{minipage}
    \hspace{0.04\textwidth}
    \begin{minipage}{0.2\textwidth}
        \centering
        \renewcommand{\arraystretch}{1.35}
        \vspace{-6pt}
        \begin{tabular}{c | c | c |}
            DAG & \makecell{nr.\\ nodes} & \makecell{nr. \\ edges} \\ 
         \hline\hline
          msnbc & 16373 & 39056 \\
         \hline
          jester & 18957 & 48220 \\
         \hline
          book & 22118 & 53647 \\
         \hline
          ad & 22277 & 47345 \\
         \hline
          cr52 & 22295 & 46968 \\
         \hline
          cwebkb & 22821 & 50251 \\
         \hline
          c20ng & 25696 & 56622 \\
         \hline
          bbc & 26023 & 54018 \\
         \hline
          kdd & 32021 & 79972 \\
         \hline
          baudio & 34567 & 94229 \\
         \hline
          $\!\!$accidents$\!$ & 37486 & $\!$102111$\!$ \\
         \hline
          \makecell{pumsb\_\\star} & 40034 & $\!$110115$\!$ \\
         \hline
          plants & 63488 & $\!$167486$\!$ \\
         \hline
          tmovie & 89389 & $\!$274214$\!$ \\
         \hline
          dna-500 & $\!$174135$\!$ & $\!$552352$\!$ \\
         \hline
        \end{tabular}
        \vspace{14pt}
    \end{minipage}
  \label{tab:psdd_inst}
\end{table}

\subsubsection*{The \texttt{sptrsv} dataset}

The task of efficiently parallelizing sparse triangular system solving (SpTRSV) can be naturally captured as a DAG scheduling problem, where each node represents a row of the triangular matrix, and each directed edge corresponds to a dependency that is established by a non-zero entry~\cite{zarebavani2022hdagg}. We again turn to the SuiteSparse matrix collection~\cite{davis2011university} at \url{https://sparse.tamu.edu/} to build the \texttt{sptrsv} dataset from such computations. We select $40$ symmetric square matrices with $1000$-$200000$ rows, mostly corresponding to application areas. We consider the lower triangular part of each of these matrices, and construct a computational DAG that represents the SpTRSV problem on the corresponding matrix. This gives us $40$ computational DAGs; $38$ of these are between $1074$ and $102158$ nodes, whereas the two largest ones have $141347$ and $146689$ nodes, respectively. The fundamental properties of these computational DAGs are summarized in Table~\ref{tab:sptrsv_inst}.

\begin{table}[t]
\caption{Properties of computational DAGs in the \texttt{sptrsv} dataset: matrix name, number of nodes and edges in the DAG.}
\hspace{-0.02\textwidth}
    \begin{minipage}{0.2\textwidth}
        \centering
        \renewcommand{\arraystretch}{1.35}
        \begin{tabular}{c | c | c |}
            matrix & \makecell{nr.\\$\!\!$nodes$\!\!$} & \makecell{nr. \\ edges} \\ 
         \hline\hline
         $\!\!$bcsstk08$\!$ & 1074 & 5943 \\
         \hline
         $\!\!$1138\_bus$\!$ & 1138 & 1458 \\
         \hline
         $\!\!$msc01440$\!$ & 1440 & 22415 \\
         \hline
         ex3 & 1821 & 25432 \\
         \hline
         $\!\!$nasa1824$\!$ & 1824 & 18692 \\
         \hline
         $\!\!$bcsstk26$\!$ & 1922 & 14207 \\
         \hline
         $\!\!$bcsstk13$\!$ & 2003 & 40940 \\
         \hline
         $\!\!$nasa2910$\!$ & 2910 & 85693 \\
         \hline
         $\!\!$bcsstk23$\!$ & 3134 & 21022 \\
         \hline
         $\!\!$bcsstk15$\!$ & 3948 & 56934 \\
         \hline
         $\!\!$sts4098$\!$ & 4098 & 34129 \\
         \hline
         $\!\!$bcsstk28$\!$ & 4410 & 107307 \\
         \hline
         $\!\!$nasa4704$\!$ & 4704 & 50026 \\
         \hline
         $\!\!\!$s3rmt3m3$\!$ & 5357 & 101169 \\
         \hline
         $\!\!$bcsstk38$\!$ & 8032 & 173714 \\
         \hline
         nd3k & 9000 & $\!\!$1635345$\!\!$ \\
         \hline
         $\!\!$msc10848$\!$ & $\!$10848$\!$ & 609465 \\
         \hline
         \makecell{Pres\_\\Poisson} & $\!$14822$\!$ & 350491 \\
         \hline
         $\!\!\!$Dubcova1$\!$ & $\!$16129$\!$ & 118440 \\
         \hline
         $\!\!$gyro\_m & $\!$17361$\!$ & 161535 \\
         \hline
         $\!\!$bodyy4 & $\!$17546$\!$ & 52196 \\
         \hline
        \end{tabular}
    \end{minipage}
    \hspace{0.04\textwidth}
    \begin{minipage}{0.2\textwidth}
        \centering
        \renewcommand{\arraystretch}{1.35}
        \begin{tabular}{c | c | c |}
            matrix & \makecell{nr.\\ nodes} & \makecell{nr. \\ edges} \\ 
         \hline\hline
         $\!\!$raefsky4$\!$ & 19779 & 654416 \\
         \hline
         $\!\!$bcsstk36$\!$ & 23052 & 560044 \\
         \hline
         $\!\!\!$msc23052$\!$ & 23052 & 565881 \\
         \hline
         smt & 25710 & $\!$1863737$\!$ \\
         \hline
         thread & 29736 & $\!$2220156$\!$ \\
         \hline
         $\!\!$ship\_001$\!$ & 34920 & $\!$2304655$\!$ \\
         \hline
         $\!\!$pdb1HYS$\!$ & 36417 & $\!$2154174$\!$ \\
         \hline
         $\!\!$cvxbqp1$\!$ & 50000 & 149984 \\
         \hline
         $\!\!$ct20stif$\!$ & 52329 & $\!$1323067$\!$ \\
         \hline
         $\!\!$Andrews$\!$ & 60000 & 350077 \\
         \hline
         $\!\!\!$Dubcova2$\!$ & 65025 & 482600 \\
         \hline
         cfd1 & 75656 & 878854 \\
         \hline
         \makecell{$\!\!$shallow\_$\!$\\water1} & 81920 & 122880 \\
         \hline
         thermal1 & 82654 & 245902 \\
         \hline
         m\_t1 & 97578 & $\!$4827996$\!$ \\
         \hline
         \makecell{$\!\!$2cubes\_\\sphere} & $\!$101492$\!$ & 772886 \\
         \hline
         \makecell{$\!\!$thermo-\\$\!\!$mech\_TC$\!$} & $\!$102158$\!$ & 304700 \\
         \hline
         $\!\!\!$bmw7st\_1$\!$ & $\!$141347$\!$ & $\!$3599160$\!$ \\
         \hline
         $\!\!\!$Dubcova3$\!$ & $\!$146689$\!$ & $\!$1744980$\!$ \\
         \hline
        \end{tabular}
        \vspace{14pt}
    \end{minipage}
  \label{tab:sptrsv_inst}
\end{table}

\section{Experimental results} \label{app:exp}

This section discusses further details of the experimental results.

\subsection{Partitioning results}

The improvements achieved via replication on our partitioning benchmarks are already discussed in the main part of the paper. Here we discuss some further technical details.

\subsubsection{Main experiments}

We emphasize again that for the ILP solver, the $P=4$ case is drastically more challenging than $P=2$; while for $P=2$ all ILPs (even replicating ones) are solved to optimality, often in a matter of minutes, for $P=4$, the replicating ILPs can often not find the optimum within the $5$-hour time limit.

We also note that in some of the hardest instances for $P=4$, the replicating ILPs actually do not find a better solution than that of the non-replicating optimum. In these cases, we considered the improvement ratio to be $1$, since this is clearly an upper bound on the actual ratio. A natural option here for future is work is to actually initialize the replicating ILPs with an (optimal) non-replicating partition; this would limit the ratio to $1$ inherently, and possibly also help guide the ILP solver to a better replicating solution. Similarly, even for the non-replicating ILP, initializing with a partitioning heuristic could be a useful way to further reduce the running time of the solver.

Recall that we have $20$ MoE related hypergraphs in our dataset, $5+5$ from Qwen and DeepSeek in both the \texttt{moe-2} and the \texttt{moe-8} datasets. It is also a natural idea to sort these into $2$ datasets of size $10$ not according to the hyperedge sizes, but the origin MoE model. Alternatively, this can also be understood as splitting the dataset into two parts according to the number of nodes. In this case, for $P=4$ we get that the mean reduction is $31.32\%$ on the Qwen hypergraphs and $35.05\%$ on the DeepSeek hypergraphs. As such, the hyperedge size indeed seems to have a stronger influence on the reduction numbers we obtain.

\subsubsection{ILP/D and ILP/R}

Table~\ref{tab:part_ilps} shows the number of cases where our two replicating ILP formulations outperform each other. The table suggests that neither of them is clearly superior: ILP/R is better in general, but it is still consistently weaker than ILP/D on \texttt{moe-8}. Hence the ideal choice of ILP might also depend on the instance properties.

\begin{table}[t]
\caption{Comparison of two ILP formulations for partitioning with replication, for $P=4$ and $\varepsilon=0.05$. For each dataset, we show the number of instances where the lower cost solution is found by ILP/D, or ILP/R, or where the costs are equal. \vspace{-8pt}}
  \begin{minipage}[b]{0.48\textwidth}
  \centering
    \renewcommand{\arraystretch}{1.45}
    \begin{tabular}{c || c | c | c | c | }
      & $\!\!$\texttt{spmv-fg}$\!\!$ & $\!\!$\texttt{spmv-rn}$\!\!$ & \texttt{moe-8} & \texttt{moe-2} \\ 
     \hline\hline
     same cost & 7 & 6 & 4 & 2 \\
     \hline
     $\!$ILP/D better$\!$ & 0 & 1 & 6 & 2 \\
     \hline
     $\!$ILP/R better$\!$ & 3 & 3 & 0 & 6 \\
     \hline
    \end{tabular}
  \end{minipage}
  \hfill
  \label{tab:part_ilps}
\end{table}

We also summarize in Table~\ref{tab:part_opt} the number of times each ILP problem is solved to optimality (out of the $10$ instances in each dataset). This confirms, firstly, that the partitioning problems with replication are indeed significantly more challenging for the ILP solver, as expected. It also shows that the second version of the replication ILP (presented in Section~\ref{sec:ilp_repl}) is generally superior to the first version (Section~\ref{sec:ilp_dupl}), especially since it also does not impose further restrictions on the number of times a node is replicated.

\begin{table}[t]
\caption{Table showing how many times the partitioning ILPs find the optimal solution within the time limit, out of the 10 instances of each dataset, with the base (non-replicating) ILP and the two replicating ILP formulations, for $P=4$ and $\varepsilon=0.05$. \vspace{-8pt}}
  \begin{minipage}[b]{0.48\textwidth}
  \centering
    \renewcommand{\arraystretch}{1.45}
    \begin{tabular}{c || c | c | c | c | }
     & $\!\!$\texttt{spmv-fg}$\!\!$ & $\!\!$\texttt{spmv-rn}$\!\!$ & \texttt{moe-8} & \texttt{moe-2} \\  
     \hline\hline
     $\!$Base ILP$\!$ & 8 & 9 & 8 & 10 \\
     \hline
     ILP/D & 3 & 0 & 0 & 0 \\
     \hline
     ILP/R & 6 & 6 & 0 & 0 \\
     \hline
    \end{tabular}
  \end{minipage}
  \hfill
  \label{tab:part_opt}
\end{table}

Our two ILP formulations with replication can also be compared in more detail based on Figure~\ref{fig:ilps_compare}, which shows the cost reduction factor with both ILPs (with the $5$ hour runtime limit) for each instance. The figure again confirms that ILP/R is generally superior, but not always: it consistently outperformed by ILP/D on the \texttt{moe-8} dataset.

\begin{figure*}[t]
    \centering
    \includegraphics[scale=0.6]{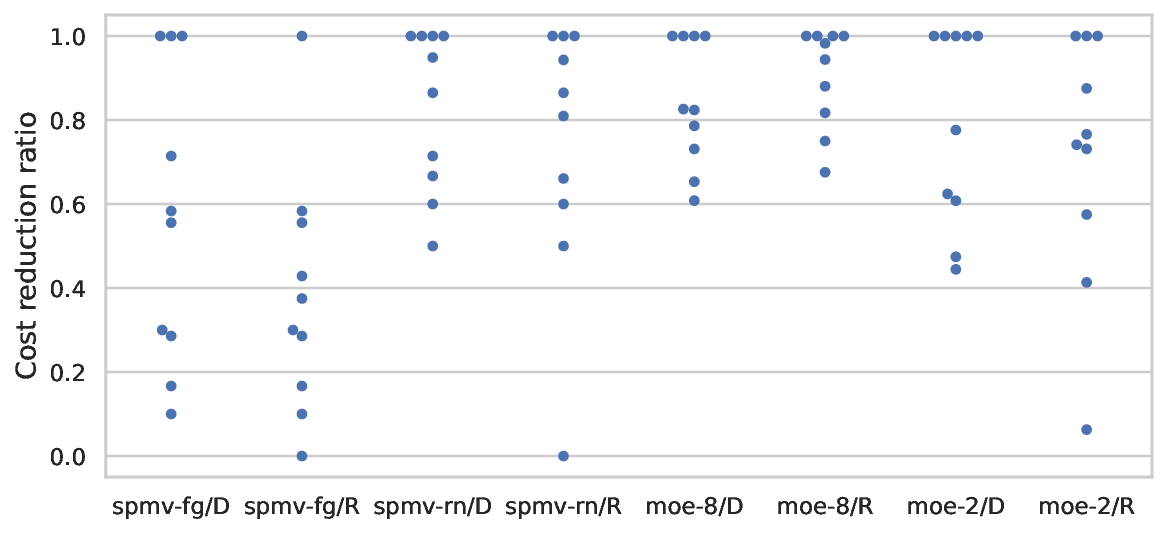}
    \vspace{-10pt}
    \caption{Illustration cost reduction ratio with replication on the different datasets for $P=4$ and $\varepsilon=0.05$. The `/D' and `/R' suffix corresponds to ILP/D and ILP/R, outlined in Sections \ref{sec:ilp_dupl} and \ref{sec:ilp_repl}, respectively.}
    \label{fig:ilps_compare}
\end{figure*}

\subsubsection{The impact of $P$ and $\varepsilon$}

For studying the impact of the parameters $P$ and $\varepsilon$, we build a smaller ablation set of $12$ hypergraphs in order to reduce the total running time of the experiments with the ILP solver. We select $3$ instances from each dataset; specifically, mesh1e1, d\_dyn and west0132 from \texttt{spmv-fg}, rajat14, impcol\_e and bcspwr04 from \texttt{spmv-rn}, and DeepSeek\_l2, DeepSeek\_l3 and DeepSeek\_l4 from \texttt{moe-8} and \texttt{moe-2}. The goal with this selection was to ensure that we get meaningful results on the same instances with a wide range of parameter combinations. Specifically, in \texttt{spmv-fg} and \texttt{spmv-rn}, we avoided the instances that already had cost $0$ with replication for $P=4$ and $\varepsilon=0.05$, since these are also separable for larger $\varepsilon$. Furthermore, we restricted ourselves to instances of at least $160$ nodes; with uniform node weights, this is the smallest $n$ which ensures that we have $\lfloor \frac{1+0.025}{2} \rfloor \cdot n > \lfloor \frac{1+0.0125}{2} \cdot n \rfloor$, and hence the set of solutions is indeed different with $P=2$ for $\varepsilon=0.025$ and $\varepsilon=0.0125$. In \texttt{moe-2}, it is only layers $2$, $3$ and $4$ of DeepSeek that have at least $160$ nodes (see the dataset details in Appendix~\ref{app:datasets}); as such, for both \texttt{moe-8} and \texttt{moe-2}, we pick the last three layers of DeepSeek.

The results for $P=2$ and $\varepsilon \in \{ 0.0125, 0.025, 0.05 \}$ are already discussed in the main part of the paper. These instances are all solved to optimality, and show clear patterns. Our goal is to conduct similar experiments for $P=4$. Our ablation set instances with $n \geq 160$ again ensure that each problem is distinct if we consider $P=4$ with the same values of $\hat{\varepsilon}=\frac{\varepsilon}{P}$; that is, if we now select $\varepsilon \in \{ 0.025, 0.05, 0.1\}$. As such, we use these three values for $P=4$. Similarly, in order to keep our preliminary experiments for $P=8$ comparable to our main experiments with $P=4$ and $\varepsilon=0.05$, we keep $\hat{\varepsilon}=\frac{\varepsilon}{P}$ fixed, and run our $P=8$ experiments with $\varepsilon=0.1$.

Note that evaluating the results of these experiments for $P=4$ (or larger) is a very difficult task for various reasons. Firstly, recall that we are not able to run our replicating ILPs to optimality on most instances; as such, when comparing the ILPs on two different parameters, the relation of the ratios is heavily dependent on how close the best ILP solutions happen to be to the optimum for this specific instance and choice of parameters. This is made even worse by the fact that we have chosen some of the largest instances for the ablation set (in order to ensure that the solution space is different for distinct $\varepsilon$ values), and these are exactly the instances that are most challenging for the ILP solver, and hence the ones were the solutions are least indicative of the true optimum values.

In the most extreme case, for $P=8$, this actually means that the replicating ILP does not find a better solution than the non-replicating ILP for any of the ablation set instances from \texttt{spmv-rn}, \texttt{moe-8} or \texttt{moe-2}. For the three ablation instances from \texttt{spmv-fg}, we find that the improvement ratios are very similar to the $P=4$ case: the reduction numbers on the same $3$ instances are $69.23\%$, $38.17\%$ and $69.23\%$ for $P=8$ and $\varepsilon=0.1$, whereas they are $71.43\%$, $44.44\%$ and $57.14\%$ for $P=4$ and $\varepsilon=0.05$, respectively.

For $P=4$, we can also compare the improvement ratios for the three cases of $\varepsilon=0.025$, $\varepsilon=0.05$, $\varepsilon=0.1$. This is simplest to do for the \texttt{spmv-rn} and \texttt{moe-2} datasets: here the non-replicating optimum is always found, the optimal costs are always above $0$, and the replicating ILPs find a better solution in all cases except $1$. For \texttt{spmv-rn}, the mean cost reduction is $17.58\%$, $25.38\%$ and $52.26\%$, whereas for \texttt{moe-2}, the mean cost reduction is $61.14\%$, $66.32\%$ and $84.26\%$, for $\varepsilon=0.025$, $\varepsilon=0.05$ and $\varepsilon=0.1$, respectively. The values clearly show an increasing tendency for larger $\varepsilon$, although we note that for \texttt{moe-2}, this is partially because there a large factor decrease for one of the instances. For the \texttt{spmv-fg} dataset, the mean cost reduction is $38.11\%$ and $59.18\%$ for $\varepsilon=0.025$ and $\varepsilon=0.05$; for $\varepsilon=0.1$, we cannot form the geomean, because, two of the instances have cost $0$ with replication (the third one has a cost reduction of $35.29\%$). Finally, for \texttt{moe-8}, there is essentially no cost decrease in these hardest instances: the mean reduction is $2.08\%$ for $\varepsilon=0.025$, and there is no decrease for the larger $\varepsilon$ values.

We also summarize these values for the $P=4$ ablation study in Table~\ref{tab:ablation_p4}. If we prefer an aggregate value over the different datasets in order to compare the different $\varepsilon$ cases, we can restrict ourselves to \texttt{spmv-rn}, \texttt{moe-2}, and the instance from \texttt{spmv-fn} that always has cost above $0$. These $7$ instances give a mean cost reduction of $42.07\%$, $49.13\%$ and $71.58\%$ for the three choices of $\varepsilon$, respectively.

\begin{table}[t]
\caption{Cost reduction with replicating ILPs for hypergraph partitioning, for $P=4$ and different choices of $\varepsilon$ for $3$ selected (large) instances from each dataset. The asterisk indicates that $2$ of the $3$ instances actually reduced the cost to $0$; the number shows the cost reduction for the third instance. \vspace{-8pt}}
  \begin{minipage}[b]{0.48\textwidth}
  \centering
    \renewcommand{\arraystretch}{1.45}
    \begin{tabular}{c || c | c | c | c | }
      & $\!\!$\texttt{spmv-fg}$\!\!$ & $\!\!$\texttt{spmv-rn}$\!\!$ & \texttt{moe-8} & \texttt{moe-2} \\ 
     \hline\hline
     $\varepsilon=0.025$ & 38.11\% & 17.58\% & 2.08\% & 61.14\% \\
     \hline
     $\varepsilon=0.05$ & 59.18\% & 25.38\% & 0\% & 66.32\% \\
     \hline
     $\varepsilon=0.1$ & 35.26\%* & 52.26\% & 0\% & 84.26\% \\
     \hline
    \end{tabular}
  \end{minipage}
  \hfill
  \label{tab:ablation_p4}
\end{table}

\subsection{Scheduling results}

The improvements achieved by our heuristics on the three main datasets are discussed in the main part of the paper. Here we discuss some further aspects and observations from the experiments.

\subsubsection{Heuristics}

In this section, we further discuss the experimental results obtained with the scheduling heuristics.

We note that between the iterations of the advanced heuristic, there are also some further technical steps to clean the schedule by e.g.\ removing communication steps that became unnecessary due to the modifications; these were omitted from the description due space constraints.  

The \texttt{hdb} dataset actually consists of three smaller parts~\cite{BSP_algos_opdas}, labeled `medium' (with $1\,000\,$--$\,2\,000$ nodes), `large' ($5\,000\,$--$\,10\,000$ nodes) and huge ($50\,000\,$--$\,100\,000$ nodes). This also offers a natural opportunity to investigate whether the size of the graphs has a significant impact on the improvements achieved by our algorithms. The improvement numbers for each subdataset for $P \! \in \! \{ 2, 4, 8, 16 \}$ are summarized in Table~\ref{tab:spaa_by_size}. The table shows rather consistent improvement numbers over the different subdatasets, which suggests that our algorithmic approach also scales acceptably to larger instances.

\begin{table}[t]
\caption{Mean cost reduction achieved by the basic and advanced heuristic on the \texttt{hdb} dataset, for $g=4$, $L=20$, organized by $P$ and instance size.}
  \begin{minipage}[b]{0.48\textwidth}
  \centering
    \renewcommand{\arraystretch}{1.45}
    \begin{tabular}{c || c | c | c | c | }
     dataset & $P=2$ & $P=4$ & $P=8$ & $P=16$ \\ 
     \hline\hline
     medium & \makecell{0.74\% /\\ 12.08\%} & \makecell{2.22\% /\\ 16.70\%} & \makecell{4.48\% /\\ 20.11\%} & \makecell{6.58\% /\\ 23.76\%} \\
     \hline
     large & \makecell{0.36\% /\\ 8.42\%} & \makecell{2.17\% /\\ 12.99\%} & \makecell{4.00\% /\\ 16.21\%} & \makecell{6.09\% /\\ 18.67\%} \\
     \hline
     huge & \makecell{0.19\% /\\ 19.89\%} & \makecell{2.34\% /\\ 22.95\%} & \makecell{3.58\% /\\ 23.19\%} & \makecell{5.87\% /\\ 27.86\%} \\
     \hline
    \end{tabular}
  \end{minipage}
  \hfill
  \label{tab:spaa_by_size}
\end{table}

Furthermore, recall that the ablation study in Table~\ref{tab:sched_ablation} shows the effect of enabling each algorithmic component of the advanced heuristic separately. For completeness, we also provide a different kind of ablation study of the same phenomenon, where we start with the basic heuristic and incrementally add the components of the advanced heuristic to this. The result of this alternative study is summarized in Table~\ref{tab:sched_ablation2}. The table again confirms that superstep merging can achieve a significant cost reduction, even after batch replication. It also shows that after we have batch replication and superstep merging, adding superstep replication sometimes even leads to slightly weaker results. This is partially an artifact of the ordering of our algorithmic components: superstep replication is applied immediately after supersteps have been merged, and thus there are several new improvement opportunities with replication; however, other methods like batch replication could remove the same communication steps with a smaller increase in computation costs. Altogether, this again confirms that batch replication and superstep merging are the two truly valuable components of the advanced heuristic.

\begin{table}[htbp]
\caption{Ablation study: cost decrease when activating the components of the heuristic separately, starting with the basic heuristic (B), then adding batch replication (BR), superstep merging (SM) and superstep replication (SR) in this order. \vspace{-8pt}}
  \begin{minipage}[b]{0.48\textwidth}
  \centering
    \renewcommand{\arraystretch}{1.45}
    \begin{tabular}{c || c | c | c || c | c | c | }
     \multirow{2}{*}{$\!\!\!\!$heuristic$\!\!\!$} & \multicolumn{3}{|c|}{$P=8$,  $g=4$,  $L=20$} & \multicolumn{3}{|c|}{$P=8$,  $g=16$,  $L=20$} \\
     \hhline{~|------|}
      & \texttt{hdb} & \texttt{psdd} & $\!$\texttt{sptrsv}$\!$ & \texttt{hdb} & \texttt{psdd} & $\!$\texttt{sptrsv}$\!$ \\ 
     \hline\hline
     B & $\!$4.11\%$\!$ & $\!$3.74\%$\!$ & $\!$6.29\%$\!$ & $\!$6.12\%$\!$ & $\!$5.03\%$\!$ & $\!$12.71\%$\!$ \\
     \hline
     $\!\!$ B+BR & $\!$11.90\%$\!$ & $\!$18.38\%$\!$ & $\!$9.56\%$\!$ & $\!$21.26\%$\!$ & $\!$28.99\%$\!$ & $\!$20.33\%$\!$ \\
     \hline
     $\!\!$ \makecell{B+BR\\+SM} & $\!$19.14\%$\!$ & $\!$23.12\%$\!$ & $\!$11.57\%$\!$ & $\!$41.81\%$\!$ & $\!$56.51\%$\!$ & $\!$26.36\%$\!$ \\
     \hline
     \makecell{$\!\!\!$B+BR\\$\!\!\!$+SM+SR$\!\!$}  & $\!$19.17\%$\!$ & $\!$23.13\%$\!$ & $\!$11.61\%$\!$ & $\!$40.97\%$\!$ & $\!$56.49\%$\!$ & $\!$26.45\%$\!$ \\
     \hline
    \end{tabular}
  \end{minipage}
  \hfill
  \label{tab:sched_ablation2}
\end{table}

\subsubsection{ILPs}

The ILP formulation for BSP scheduling has already been introduced in~\cite{BSP_DAG_opdas} and experimentally studied in~\cite{BSP_algos_opdas}. The formulation relies on binary variables $x_{v,p,s}$ for all nodes $v$, processors $p$ and supersteps $s$, indicating whether $v$ is computed on $p$ in superstep $s$. Other key variables express whether $v$ is already present on $p$ in superstep $s$, and whether $v$ is sent from $p_1$ to $p_2$ in superstep $s$. The number of variables is thus much larger than in the non-replicating hypergraph partitioning ILP: there is also a time dimension $s$ in the indices, and the communication variables scale with $P^2$. This already indicates why scheduling is much more challenging as an ILP than partitioning.

The generalization of this ILP to replication is surprisingly simple. The original formulation has the constraint $\sum_{p} \sum_{t} x_{v,p,t} =1$ for every node $v$; we replace this by $\sum_{p} \sum_{t} x_{v,p,t} \geq 1$.

We use this ILP formulations to also derive optimal schedules on very small computation DAGs with COPT. Due to the small DAG sizes ($40$--$80$ nodes), we only consider $P=2$ and $P=4$ here. We pick $g=4$ as in our main experiments, but scale down $L$ to $5$, because our preliminary experiments showed that the original $L=20$ is too large to parallelize these tiny graphs; the optimum is very often a sequential schedule on a single processor.

We use a shorter time limit of $2$ hours for the ILPs in these experiments. This allows to solve $8$ of the instances to optimality for $P=2$, and $4$ of them for $P=4$; this already gives a good indication of the general performance of computational DAGs of this size. The mean cost reductions for $P=2$ and $P=4$ on these instances is $12.33\%$ and $21.96\%$. In order to expand this dataset, we can also include the instances where the solver could prove that the gap between the best found solution and the best proven lower bound is at most $25\%$. In our experience, the best solution cost is rather unlikely to change after this point in the solving process, and hence we can also consider these costs `essentially optimal'. For $P=2$ and $P=4$, respectively, this relaxation increases the number of instances to $12$ and $6$, and modifies the mean cost reduction numbers to $12.99\%$ and $21.08\%$. These are almost identical to the same reduction values in the optimal set.

The improvement numbers above are very similar to those obtained with our advanced scheduling heuristic on the \texttt{hdb} and \texttt{psdd} datasets with comparable parameters. This suggests that the heuristic scheduling costs in our experiments might also be a good indicator of the cost ratio of the actual optimum costs with and without replication. The data also shows that in this small dataset, the optimal non-replicating solution in this small dataset always consists of more than $1$ superstep, so these are indeed non-trivial, parallel schedules, thus confirming that this is a reasonable choice of parameters.

\end{document}